\documentclass[10pt,a4paper]{article}

\usepackage[utf8]{inputenc}
\usepackage[english]{babel}
\usepackage{hyperref}
\usepackage{amsmath,amsthm,enumerate}
\usepackage{amssymb}
\usepackage{graphicx}
\usepackage{tikz,pgf}
\usepackage[hmargin=2.3cm,vmargin=2.8cm]{geometry}
\usepackage{subfig}
\usetikzlibrary{positioning}

\newtheorem{theorem}{Theorem}

\newtheorem{proposition}[theorem]{Proposition}
\newtheorem{lemma}[theorem]{Lemma}
\newtheorem{observation}[theorem]{Observation}

\newtheorem{corollary}[theorem]{Corollary}

\newtheorem{definition}[theorem]{Definition}

\newtheorem{conjecture}[theorem]{Conjecture}

\newtheorem{claim}{Claim}[theorem]

\newcommand{\smallqed}{{\tiny $\left(\Box\right)$}}
\newcommand{\claimproof}{\noindent\emph{Proof of claim.} }

\newcommand{\sharptwo}{^{[\sharp 2]}}
\newcommand{\Omegasharp}{\omega\sharptwo}
\newcommand{\Chisharp}{\chi\sharptwo}

\usetikzlibrary{shapes.geometric}

\begin{document}

		\title{Exact square coloring of subcubic planar graphs}

		\author{Florent Foucaud\footnote{Univ. Bordeaux, CNRS, Bordeaux INP, LaBRI, UMR5800, F-33400 Talence, France.}~\footnote{Univ. Orléans, INSA Centre Val de Loire, LIFO EA 4022, F-45067 Orléans Cedex 2, France.} \and Hervé Hocquard\footnotemark[1] \and Suchismita Mishra\footnote{Department of Mathematics, IIT Madras, Chennai 600036, India.} \and Narayanan Narayanan\footnotemark[3] \and Reza Naserasr\footnote{Université de Paris, CNRS, IRIF, F-75006, Paris, France.} \and Éric Sopena\footnotemark[1] \and Petru Valicov\footnote{Aix-Marseille Univ, Université de Toulon, CNRS, LIS, Marseille, France.}~~\footnote{LIRMM, CNRS, Université de Montpellier, France.}}

\maketitle

\begin{abstract}
  We study the exact square chromatic number of subcubic planar graphs. An exact square coloring of a graph $G$ is a vertex-coloring in which any two vertices at distance exactly~$2$ receive distinct colors. The smallest number of colors used in such a coloring of $G$ is its exact square chromatic number, denoted $\chi\sharptwo(G)$. This notion is related to other types of distance-based colorings, as well as to injective coloring. Indeed, for triangle-free graphs, exact square coloring and injective coloring coincide. We prove tight bounds on special subclasses of planar graphs: subcubic bipartite planar graphs and subcubic $K_4$-minor-free graphs have exact square chromatic number at most~$4$. We then turn our attention to the class of fullerene graphs, which are cubic planar graphs with face sizes~$5$ and~$6$. We characterize fullerene graphs with exact square chromatic number~$3$. Furthermore, supporting a conjecture of  Chen, Hahn, Raspaud and Wang (that all subcubic planar graphs are injectively $5$-colorable) we prove that any induced subgraph of a fullerene graph has exact square chromatic number at most 5. This is done by first proving that a minimum counterexample has to be on at most 80 vertices and then computationally verifying the claim for all such graphs.
\end{abstract}

\section{Introduction}

In this paper, we study exact distance coloring problems for graphs. The celebrated Hadwiger-Nelson problem asking for the ``chromatic number of the plane'' falls into this category of problems: there, one wishes to assign a color to each point of the Euclidean plane, such that two points at distance exactly~$1$ receive distinct colors. A recent breakthrough result on this problem, showing that at least five colors are necessary, appeared in~\cite{deGrey} (it is long known that seven colors suffice~\cite{hadwiger}). Similar problems are studied for other metric spaces, see~\cite{K15}. In the graph setting, for a positive integer~$p$, an \emph{exact $p$-distance coloring} of a graph $G$ is an assignment of colors to the vertices of $G$, such that two vertices at distance exactly~$p$ receive distinct colors~\cite[Section 11.9]{sparsity-book}. The \emph{exact $p$-distance chromatic number} of $G$, denoted $\chi^{[\sharp p]}(G)$, is the smallest number of colors in an exact $p$-distance coloring of $G$ and the study of this parameter is gaining growing attention, see~\cite{BEHJ19,HKQ16,K15,Q17}. Denoted $G^{[\sharp p]}$, the exact distance-$p$-power of $G$, is a graph obtained by taking vertices of $G$ and adding an edge between any two distinct vertices at distance exactly~$p$ in $G$. We have $\chi^{[\sharp p]}(G)=\chi(G^{[\sharp p]})$. Thus, for $p=1$ this notion coincides with the usual chromatic number.
Similarly denote by $\omega^{[\sharp p]}(G)$ and $\alpha^{[\sharp p]}(G)$, respectively, the clique number and the independence number of $G^{[\sharp p]}$, that is $\omega^{[\sharp p]}(G)=\omega(G^{[\sharp p]})$ and $\alpha^{[\sharp p]}(G) = \alpha(G^{[\sharp p]})$.

Exact distance $p$-powers of graphs were first studied by Simi\'c~\cite{simic}, see also~\cite{azimi,products} for more recent works. Exact $p$-distance colorings have first been studied for graphs of bounded expansion. For a fixed graph class $\mathcal C$ of bounded expansion (for example, the class of planar graphs), the exact $p$-distance chromatic number is bounded by an absolute constant for graphs in $\mathcal C$ when $p$ is odd~\cite{HKQ16,sparsity-book}, where the constant is determined by the class $\mathcal{C}$ and $p$, and by a linear function of the maximum degree when $p$ is even~\cite{HKQ16} where the coefficients of the linear function are determined by the class $\mathcal{C}$ and $p$. Exact $p$-distance colorings have been studied in specific graph classes: trees~\cite{BEHJ19}, graphs of bounded tree-width~\cite{HKQ16}, chordal graphs~\cite{Q17}, graphs of bounded genus~\cite{HKQ16,Q17}.

Distance-based colorings of graphs have been extensively studied since the first papers on the subject were published in the 1960s by Kramer and Kramer~\cite{KK69A,KK69B}. In their setting, for a positive integer $p$, a \emph{$p$-distance coloring} of a graph $G$ is an assignment of colors to the vertices of $G$, such that two vertices at distance \emph{at most}~$p$ receive distinct colors. The smallest possible number of colors used in such a coloring is denoted by $\chi^p(G)$ (see~\cite{AM02,HHMR07,T18,W77} for some important results). The problem of $2$-distance coloring subcubic and cubic planar graphs is already far from trivial, and the focus of many research works: see~\cite{BI12,FHS18,H09,HHMR07,T18}. Notably, a special case of a conjecture by Wegner~\cite{W77}, recently solved in~\cite{T18}, states that $\chi^2(G)\leq 7$ for every subcubic planar graph $G$.

The goal of the present paper is to study \emph{exact} $p$-distance colorings with a focus on the special case $p=2$ and for subclasses of planar graphs. Given a graph $G$, the graph $G\sharptwo$ is called the \emph{exact square} of $G$. Similarly, an exact $2$-distance coloring of $G$ is called an \emph{exact square coloring} of $G$, and $\chi\sharptwo(G)$ is the \emph{exact square chromatic number} of $G$.  

What makes the study of exact distance coloring rather more difficult, is that $\chi^{[\sharp p]}$ is not necessarily monotone with respect to taking subgraphs. Indeed $K_n$ is exact $p$-distance 1-colorable for every $p\geq 2$. However, if restricted to exact square coloring ($p=2$), then $\Chisharp$ is monotone with respect to taking induced subgraphs. On the other hand, the parameter $\chi\sharptwo$ is unbounded even for trees, indeed for the star $K_{1,t}$ with $t$ leaves, $(K_{1,t})\sharptwo$ is isomorphic to the disjoint union of the complete graphs $K_t$ and $K_1$, and hence $\chi\sharptwo(K_{1,t})=t$. Thus, it is natural to restrict the study of exact square colorings to graphs with no induced $K_{1,t}$. For triangle-free graphs, this turns out to be the same as bounding the maximum degree. Thus, the class of subcubic graphs will be a focus of this study.

A related notion is the one of \emph{injective coloring}, introduced in~\cite{HKSS02} and well-studied since then, see for example~\cite{BELLMW2017,CHRW2012,LST09}. An injective coloring of a graph $G$ is a vertex-coloring where any two vertices that are joined by a path of length~$2$ receive distinct colors. The smallest number of colors used in an injective coloring of $G$ is its \emph{injective chromatic number}, denoted $\chi_i(G)$.

From these definitions, we have the following inequalities for any graph $G$:
$$\chi\sharptwo(G)\leq \chi_i(G)\leq \chi^2(G).$$

Moreover, whenever $G$ is triangle-free, we have $\chi\sharptwo(G)=\chi_i(G)$, since in $G$, vertices joined by a path of length~$2$ must also be at distance~$2$. Since many of our results are for triangle-free graphs, they can be re-interepreted as results on injective colorings. For a subcubic graph $G$, the maximum degree of $G\sharptwo$ is at most~$6$ and thus $\chi\sharptwo(G)\leq 7$. In fact, the bound $\chi_i(G)\leq 7$ is also true~\cite{HKSS02}. Equality in the latter is shown to hold if and only if $G$ is the Heawood graph~\cite{CHRW2012,HKSS02}, and using the arguments of~\cite{CHRW2012,HKSS02}, the same holds for the exact square chromatic number. Thus, for every subcubic planar graph $G$, we have $\chi_i(G)\leq 6$. A stronger bound was conjectured as follows.

\begin{conjecture}[\cite{CHRW2012}]\label{conj}
If $G$ is a subcubic planar graph, then $\chi_i(G)\leq 5$.  
\end{conjecture}

Note that there is a subcubic planar graph $G$ (with triangles) satisfying $\chi_i(G)=5$~\cite{CHRW2012}, so if true, the conjectured bound would be tight. However, we do not know of any subcubic planar graph with exact square chromatic number~$5$. Conjecture~\ref{conj} was generalized to arbitrary values of the maximum degree~\cite{CHRW2012,LS15}, in the spirit of a well-studied conjecture by Wegner~\cite{W77} on the (non-exact) square chromatic number. Conjecture~\ref{conj} was proved for $K_4$-minor-free graphs~\cite{CHRW2012} (and it follows from~\cite{BKTL04} that $\chi\sharptwo(G)\leq 5$ for any $K_4$-minor-free graph $G$). Furthermore, the bound was improved to four colors for outerplanar subcubic graphs~\cite{MO18}. If $G$ is a subcubic graph with girth at least~$19$ (resp. $10$) then $\chi_i(G)\leq 3$ (resp. $\chi_i(G)\leq 4$)~\cite{LST09}. If $G$ has girth at least~$6$, then $\chi_i(G)\leq 5$~\cite{BELLMW2017}. Similar (but larger) bounds are known for $\chi^2(G)$, see~\cite{BI12,FHS18}.

Exact square colorings also appear in another, more general, context: the one of $L(p,q)$-labelings. Given two non-negative integers $p$ and $q$, an $L(p,q)$-labeling of a graph $G$ is an assignment $\ell$ of non-negative integers to the vertices of $G$, such that for any two vertices $u$ and $v$, we have  $|\ell(u)-\ell(v)|\geq p$ if $u$ and $v$ are adjacent, and $|\ell(u)-\ell(v)|\geq q$ if they are at distance~$2$. The first (and most) studied case is when $p=2$ and $q=1$~\cite{GY92}; see the survey~\cite{Lpqsurvey}. Thus, for any graph $G$, there is a $1$-to-$1$ correspondence between: (a) $L(1,0)$-labelings and classic vertex-colorings of $G$, (b) $L(1,1)$-labelings and (non-exact) square colorings of $G$, and (c) $L(0,1)$-labelings and exact square colorings of $G$. However, it seems that $L(0,1)$-labelings are rarely studied, see~\cite{BKTL04,Lpqsurvey} for a few references.

We first prove some results for specific classes of subcubic planar graphs in Section~\ref{sec:general}, filling some gaps from the literature. Let $G$ be a subcubic graph. We show that $\chi\sharptwo(G)\leq 4$ if $G$ is $K_4$-minor-free. We also show that if $G$ is planar and bipartite, then its exact square is planar, thus $\chi_i(G)=\chi\sharptwo(G)\leq 4$. Moreover these bounds are tight: there exist bipartite subcubic $K_4$-minor-free graphs with exact square chromatic number~$4$. In passing we also show that the exact square of every subcubic bipartite outerplanar graph $G$ is outerplanar, thus $\chi_i(G)=\chi\sharptwo(G)\leq 3$. This is tight since for every tree $T$ with a vertex of degree~$3$, $\chi\sharptwo(T)\geq 3$.

The main focus of our paper (Section~\ref{sec:fullerenes}) is on fullerene graphs, which are cubic planar graphs where every face has length~$5$ or~$6$. They form an important and interesting class of cubic graphs,
whose definition arises from chemistry: indeed they correspond to the structure of fullerene molecules. Their graph-theoretic properties are well-studied, in particular, in relation with colorings. See the survey~\cite{fullerenes-survey}. For instance since it is known that fullerene graphs have girth~$5$, any exact square coloring of a fullerene graph is also an injective coloring, and so our results for this class apply to both settings. We first characterize in Section~\ref{sec:fullerenes-3col} those fullerene graphs with exact square chromatic number~$3$. It turns out that these fullerene graphs are a special class of so-called \emph{$(6,0)$-nanotubes}~\cite{nanotubes,KS08} which we call \emph{drums}. Then, in Section~\ref{sec:fullerenes-5col}, we prove Conjecture~\ref{conj} for fullerene graphs. The proof is computer-assisted: we first consider a potential minimum counterexample, and prove that it cannot have more than 80 vertices. We then use a computer program to check the list of fullerene graphs of order up to 80, which is available online and certified complete~\cite{BD97,BGM12,GM15}.

We conclude in Section~\ref{sec:conclu}.

\section{Generalities}\label{sec:general}

We now recall some facts from the literature and prove a few new results.

\subsection{Preliminaries}

Given a vertex $v$, we denote by $d(v)$ the degree of $v$ and we say that $v$ is a $k$-vertex if $d(v)=k$. For a planar graph with a given planar embedding, we call a $k$-face a face of length $k$.

A \emph{thread} in a graph $G$ is a path all of whose internal vertices are 2-vertices.
We will use the following lemma.

\begin{lemma}[\cite{NRS97}]\label{lem:general_girth_threads}\label{lemm:threads}
Any planar graph $G$ of girth at least $5d+1$ ($d\geq 1$) contains a vertex of degree~$1$ or a thread with at least $d$ internal vertices.
\end{lemma}

\begin{observation}
  For any graph $G$, $G\sharptwo$ is a subgraph of $\overline{G}$ (the complement of $G$). If $G$ has diameter~$2$, then $G\sharptwo$ is isomorphic to $\overline{G}$.
\end{observation}

The two following results were formulated in the context of injective coloring~\cite{CHRW2012,HKSS02}, but the same arguments hold for exact squares.

\begin{theorem}[\cite{CHRW2012,HKSS02}]\label{thm:projective}
If $G$ is a connected graph of maximum degree $\Delta=k+1$, then $\Omegasharp(G)\leq k^2+k+1$. Moreover, equality may only happen if $G$ is the incidence graph of a projective geometry of order $k$, in which case $G\sharptwo$ is isomorphic to two copies of $K_{k^2+k+1}$.
\end{theorem}

Using Brook's theorem, the following can be deduced for the special case $k=2$.

\begin{corollary}[\cite{CHRW2012}]\label{prop:6-col}
  Let $G$ be a connected subcubic graph. Then, $\Chisharp(G)\leq 6$, unless $G$ is the Heawood graph (in which case $G\sharptwo$ is isomorphic to two disjoint copies of $K_7$).
\end{corollary}

It is not difficult to find general subcubic graphs whose exact square contains a $5$-clique (e.g. the Petersen graph\footnote{If $G$ has diameter $2$, then $G\sharptwo$ is isomorphic to $\overline{G}$ and if $G$ has girth $5$, then $\alpha(G\sharptwo)=2$.}) or a $6$-clique (e.g. the triplex graph, a cubic graph of girth 5 and order 12, depicted in Figure~\ref{subfig:triplex}).

Another interesting graph $G$ with $\Chisharp(G)=6$ is the subcubic bipartite graph of order $22$ built in a similar fashion as the Heawood graph, see Figure~\ref{subfig:bip_22_heawood_style}. This graph is also the incidence graph of the $11_3$ geometric configuration number $31$ described in~\cite{PD84}. Each connected componenent of its exact square is isomorphic to the (non-exact) $3$-distance power of the 11-cycle $C_{11}$. It has maximum degree~$6$, clique number~$4$, chromatic number~$6$ and is thus an extremal example for Reed's conjecture stating that $\chi(G)\leq \left\lceil\frac{\omega(G)+\Delta(G)+1}{2}\right\rceil$, for every graph $G$~\cite{reed}.

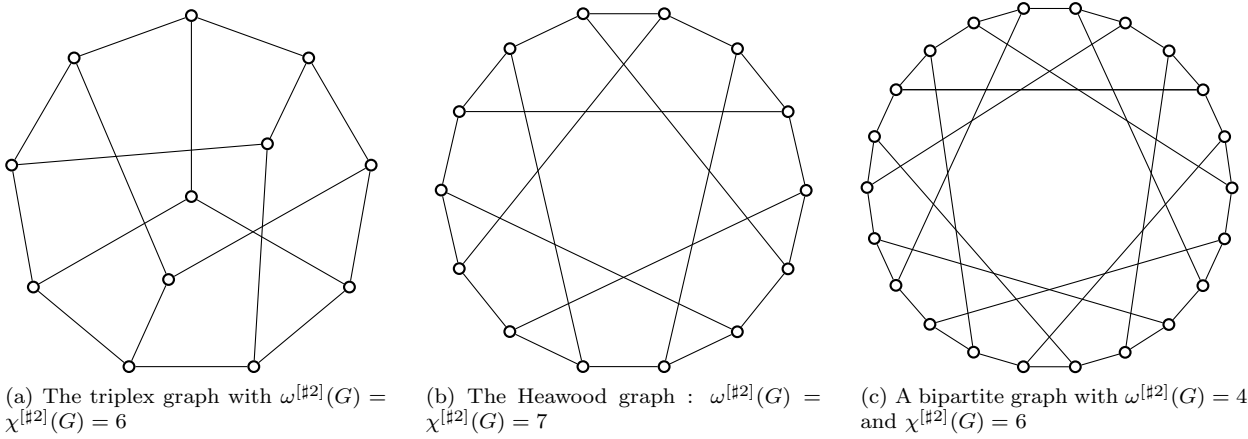
\begin{figure}[!ht]
\centering
\subfloat[The triplex graph with $\Omegasharp(G)=\Chisharp(G)=6$]{\label{subfig:triplex}

\begin{tikzpicture}[scale=1]
\tikzstyle{vertex}=[circle,draw,thick,fill=white,inner sep=0.5mm]

	\node[vertex] (o) at (0,0) {};
	\node[vertex] (o') at (1,0.7) {};
	\node[vertex] (o'') at (-0.3,-1.1) {};

	\foreach \i in {1,...,9}{
	    \node[vertex] (u\i) at (40*\i+10:2.4) {};
	    	}
	 
	\draw (u1) -- (u2) -- (u3) -- (u4) -- (u5) -- (u6) -- (u7) -- (u8) -- (u9) -- (u1);   	
	
	\draw (u2) -- (o) -- (u5);
	\draw (o) -- (u8);
	
	\draw (u1) -- (o') -- (u4);
	\draw (o') -- (u7);
	
	\draw (u3) -- (o'') -- (u6);
	\draw (o'') -- (u9);
	    	
\end{tikzpicture}
}\hspace*{0.4cm}
\subfloat[The Heawood graph : $\Omegasharp(G)=\Chisharp(G)=7$]{\label{subfig:heawood}
\begin{tikzpicture}[scale=1]
\tikzstyle{vertex}=[circle,draw,thick,fill=white,inner sep=0.5mm]

	\foreach \i in {1,...,14}{
	    \node[vertex] (u\i) at (25.71*\i:2.4) {};
	}
	
	\draw (u1) -- (u2) -- (u3) -- (u4) -- (u5) -- (u6) -- (u7) -- (u8) -- (u9) -- (u10) -- (u11) -- (u12) -- (u13) -- (u14) -- (u1);	
	 
	\draw (u1) -- (u6);
	\draw (u2) -- (u11);
	\draw (u3) -- (u8);
	\draw (u4) -- (u13);
	\draw (u5) -- (u10);
	\draw (u7) -- (u12);
	\draw (u9) -- (u14);
			    	
\end{tikzpicture}
}\hspace*{0.4cm}
\subfloat[A bipartite graph with $\Omegasharp(G)=4$ and $\Chisharp(G)=6$]{\label{subfig:bip_22_heawood_style}

\begin{tikzpicture}[scale=1]
\tikzstyle{vertex}=[circle,draw,thick,fill=white,inner sep=0.5mm]

	\foreach \i in {1,...,22}{
	    \node[vertex] (u\i) at (16.36*\i:2.4) {};
	}
	
	\draw (u1) -- (u2) -- (u3) -- (u4) -- (u5) -- (u6) -- (u7) -- (u8) -- (u9) -- (u10) -- (u11) -- (u12) -- (u13) -- (u14) -- (u15) -- (u16) -- (u17) -- (u18) -- (u19) -- (u20) -- (u21) -- (u22) -- (u1);	
	 
	\draw (u1) -- (u16);
	\draw (u2) -- (u9);
	\draw (u3) -- (u18);
	\draw (u4) -- (u11);
	\draw (u5) -- (u20);
	\draw (u6) -- (u13);
	\draw (u7) -- (u22);
	\draw (u8) -- (u15);
	\draw (u9) -- (u2);
	\draw (u10) -- (u17);
	\draw (u12) -- (u19);
	\draw (u14) -- (u21);

\end{tikzpicture}
}
\caption{Some extremal examples of cubic graphs having high exact square chromatic number.}
\label{fig:tightcubicgraphs}
\end{figure}

The following lemma is interesting to observe.

\begin{lemma}\label{lemm:IS}
Let $G$ be a subcubic planar graph and let $I$ be an independent set of $G$. Then, $G\sharptwo[I]$ is planar.
\end{lemma}
\begin{proof}
Let $G'$ be the graph built on $I$ as follows: for each vertex $x \in V(G)\setminus I$, delete $x$ and join all its neighbors in $I$. Observe that, as $x$ has at most three neighbors in $G$, and by considering a plane embedding of $G$, the new graph $G'$ is planar. However this graph may have multiedges (when two vertices of $I$ are in a 4-cycle). By removing all but one edge from any group of parallel edges we get $G\sharptwo[I]$.
\end{proof}

We deduce the following from Lemma~\ref{lemm:IS}.

\begin{proposition}\label{prop:OmegasharpPlanar}
For any subcubic planar graph $G$, we have $\Omegasharp(G)\leq 4$.
\end{proposition}
\begin{proof}
  Let $K$ be a clique in $G\sharptwo$. In $G$, $K$ is an independent set. Thus, by Lemma~\ref{lemm:IS}, $G\sharptwo[K]$ is a planar complete graph, which implies that $|K|\leq 4$, as desired.
\end{proof}

The bound of Proposition~\ref{prop:OmegasharpPlanar} is sharp, as we see next.

\begin{proposition}\label{prop:OmegasharpSP}
  There exist bipartite subcubic $K_4$-minor-free graphs $G$ of girth $6$ with $\Omegasharp(G)=4$.
\end{proposition}
\begin{proof}
Consider the graph $G$ consisting of two vertices connected with three vertex-disjoint paths of length 3. Then $G\sharptwo$ contains two disjoint copies of $K_4$.
\end{proof}

\subsection{Bipartite planar graphs}

We deduce the following from Lemma~\ref{lemm:IS}.

\begin{theorem}\label{thm:bipartite}
Let $G$ be a connected bipartite planar subcubic graph. Then, $G\sharptwo$ consists of two planar connected components, each induced by one part of the bipartition of $G$. Moreover, if $G$ is outerplanar, then these two components are outerplanar.
\end{theorem}
\begin{proof}

From Lemma~\ref{lemm:IS} it follows that each part of $G$ induces a planar graph in $G\sharptwo$. That each part induces a connected subgraph of $G\sharptwo$ is a consequence of $G$ being connected. Indeed if $x$ and $y$ are in the same part of $G$, an $x-y$ path $P$ in $G$ will induce an $x-y$ path in $G\sharptwo$ based on the vertices of $P$ that are in the same part as $x$ and $y$. 
  
For the second part, suppose that $G$ is outerplanar with bipartition $(X,Y)$. Recall that $G\sharptwo$ consists of two connected components: the one induced by $X$ and the one induced by $Y$. It suffices to prove that each biconnected component of $G\sharptwo$ is outerplanar. 

If $G$ has a bridge, then the two vertices of this bridge are cut-vertices in $G\sharptwo$. Since they belong to two different connected components of $G\sharptwo$, each component contains a cut-vertex. Since we consider biconnected components of $G\sharptwo$, we may assume in the following that $G$ has no bridge.

 Since $G$ is cubic, we conclude that $G$ is biconnected. Consider an outerplanar embedding of $G$ and let $C=v_1v_2v_3\ldots v_{2k}$ be the facial cycle of this embedding. Then $X=\{v_1,v_3, \ldots, v_{2k-1}\}$ and $Y=\{v_2,v_4, \ldots, v_{2k}\}$ is the bipartition of $G$. As $G$ is bipartite, $v_i$ is not adjacent to $v_{i+2}$ (where the sum in the indices is taken $\pmod 2$). Therefore, $C_X=v_1v_3 \ldots v_{2k-1}$ is a Hamiltonian cycle of the component of $G\sharptwo$ induced by $X$ and $C_Y=v_2v_4 \ldots v_{2k}$ is a Hamiltonian cycle of the component of $G\sharptwo$ induced by $Y$. We show that $C_X$ can be the outer cycle of a planar embedding of the two components of $G\sharptwo$ induced by $X$ and $Y$. (The claim for $C_Y$ is analogous.) The edges of $C_X$ can be presented on the outer face of $G$ following the cyclic order of $C$. If a vertex $v_i\in X$ is adjacent to a vertex $v_j$, $j\not\in \{i-1, i+1\}$, (i.e., $v_iv_j$ is a chord of $C$), then we draw two edges along the paths $v_iv_jv_{j-1}$ and $v_iv_jv_{j+1}$. We observe that, since $G$ is bipartite, $v_j\in Y$ and this operation only applies to $v_i$. Furthermore, since $G$ is subcubic, these newly drawn edges would not cross. This embedding then is clearly an outerplanar embedding, with $C_X$ being its outer face. It remains to show that it has all the edges of $G\sharptwo$ induced by the $X$ part of $G$. Let $v_k$ and $v_l$ be two vertices of $X$ with a common neighbor $v_p$ (thus $v_p\in Y$). Since $G$ is cubic, $v_p$ is either a neighbor of $v_k$ on $C$ or a neighbor of $v_l$ on $C$ (otherwise $v_p$ is of degree at least $4$). If it is a common neighbor of both of them, then $v_k$ and $v_l$ are consecutive vertices of $C_X$ and adjacent in our embedding. Otherwise, without loss of generality, we may assume that $v_p$ is next to $v_k$ in $C$ and that $v_pv_l$ is a chord of $C$. Then we have drawn an edge along the path $v_lv_pv_k$.  
  \end{proof}

Theorem~\ref{thm:bipartite}, together with the Four-Color Theorem, implies that the exact square of a bipartite subcubic planar graph $G$ is 4-colorable. If $G$ is, furthermore, assumed to be outerplanar, then its exact square, being outerplanar, is 3-colorable.

\subsection{$K_4$-minor-free graphs}\label{sec:SP}
A \emph{nested ear decomposition} of a graph $G$ is a partition of the edges of $G$ into $E_1,\ldots,E_k$ (the ears), such that the following conditions hold.
\begin{itemize}
	\item[(i)] For every ear $E_i$, only the two end-vertices might be the same (thus $E_i$ induces a path or a cycle).
	\item[(ii)] For every ear $E_i$ with $i>1$, there is an ear $E_j$ with $j<i$ such that the two endpoints of $E_i$ belong to $E_j$ (we say that $E_i$ is \emph{nested} in $E_j$, and the sub-path of $E_j$ between the two endpoints of $E_i$ is the \emph{nest interval} of $E_i$).
	\item[(iii)] Apart from the endpoints of $E_i$, for every $j<i$, no other vertex of $E_i$ belongs to $E_j$.
	\item[(iv)] If two ears $E_i$ and $E_{i'}$ are both nested in $E_j$, then their nest intervals are either disjoint or one is contained in the other.
\end{itemize}

This concept was defined in~\cite{eardecSP}. It is known that in any $K_4$-minor-free graph, every biconnected component has a nested ear decomposition. Indeed, a graph is $K_4$-minor-free if and only if every biconnected component is two-terminal series-parallel~\cite{bod}, and every biconnected two-terminal series-parallel graph has a nested ear decomposition~\cite{eardecSP}.

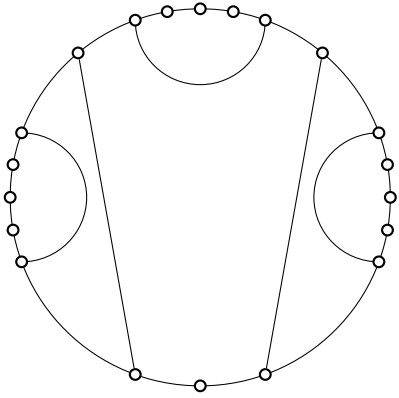
\begin{figure}[!ht]
\centering
	\begin{tikzpicture}[scale=0.25]
	\tikzstyle{vertex}=[circle,draw,thick,fill=white,inner sep=0.5mm]
	
	\draw (0,0) circle(10);
	\draw (-3.42,9.397) arc(-180:0:3.42);
	\draw (-9.397,-3.42) arc(-90:90:3.42);
	\draw (9.397,-3.42) arc(-90:-270:3.42);
	
	\foreach \i in {1,...,5}{
	    \node[vertex] (u\i) at (10*\i+60:10) {};
	}
	
	\node[vertex] (u0) at (50:10) {};
	\node[vertex] (u6) at (130:10) {};

	\foreach \i in {7,...,11}{
	    \node[vertex] (u\i) at (10*\i+90:10) {};
	}

	\node[vertex] (u12) at (250:10) {};
	\node[vertex] (u13) at (270:10) {};
	\node[vertex] (u14) at (290:10) {};
	
	\draw (u6) -- (u12);
	\draw (u0) -- (u14);
	
	\foreach \i in {15,...,19}{
	    \node[vertex] (u\i) at (10*\i+190:10) {};
	}
	
	\end{tikzpicture}
\caption{An outerplanar graph with $\Delta=3$ and $\Chisharp(G)=4$~\cite{MO18}.}\label{fig:outerplanar}
\end{figure}

\begin{theorem}\label{thm:SP-4-col}
If $G$ is a subcubic $K_4$-minor-free graph, then $\Chisharp(G)\leq 4$. Furthermore, this bound is tight on
subcubic bipartite $K_4$-minor free graphs of girth at least~6, and on outerplanar graphs.
\end{theorem}
\begin{proof}
	
Let $G$ be a minimum counterexample, that is, $G$ is subcubic and $K_4$-minor-free, $\chi\sharptwo(G)>4$ and the exact square of every smaller subcubic $K_4$-minor-free graph is $4$-colorable.

\begin{claim}\label{clm:2-connected-K4minorfree}
$G$ is $2$-connected.
\end{claim}
\claimproof
Suppose that $G$ has a cut-vertex $v$. Since $G$ is subcubic, then it has a bridge $xy$. The case where one of the two vertices (say $x$) has degree $1$ is easy, since any exact square 4-coloring of $G-x$ would extend to $G$.

So assume that both $x$ and $y$ are vertices of degree at least 2. We remove the edge $xy$ from $G$: this creates two components. We add a degree 1 neighbor $y'$ to $x$ and a degree 1 neighbor $x'$ to $y$. Let $G_x$ and $G_y$ be the components containing $x$ and $y$, respectively. Now, let $\phi$ be a $k$-coloring of $G_x\sharptwo$ and $G_y\sharptwo$.

  Suppose first that we have $\phi(x)=\phi(y')$ and $\phi(x')=\phi(y)$. Then, we
  interchange the colors $\phi(x)$ and $\phi(y)$ in $G_y$.
  This new coloring induces a valid $k$-coloring of $G\sharptwo$.

  Similarly, if $\phi(x)\neq \phi(y')$ and $\phi(x')\neq \phi(y)$, we permute the colors of vertices of $G_y$ so that $\phi(x)=\phi(x')$ and  $\phi(y)=\phi(y')$. Again this yields a valid $k$-coloring of $G\sharptwo$.

  Finally, suppose that $\phi(x)=\phi(y')$ but $\phi(x')\neq \phi(y)$ (the symmetric case is handled similarly). Since $k\geq 4$, there is one free color $c_f$ among the neighbors of $x$ (in $G$). We now permute the colors in $G_y$ such that $\phi(x')=\phi(x)$ and $\phi(y)=c_f$. Again this yields a valid $k$-coloring of $G\sharptwo$.
\hfill\smallqed\medskip

\begin{claim}\label{clm:2-vertex-triangle-4cycle}
$G$ contains no $2$-vertex lying on a $k$-cycle, with $k\in\{3,4\}$.
\end{claim}
\claimproof
Suppose $v$ is a 2-vertex in $G$ lying on a $k$-cycle ($k\in\{3,4\}$) and consider $G'=G-v$. Then $\Chisharp(G')\leq 4$ by minimality of $G$. Now, since $v$ is part of a $k$-cycle in $G$, it is easy to see that the distance in $G'$ between any pair of vertices $x,y$ is the same as their distance in $G$. Thus an exact square coloring of $G'$ is valid in $G$. Since $G$ is subcubic and $v$ is lying on $k$-cycle, we conclude that $v$ has degree at most 3 in $G\sharptwo$. Therefore $v$ can be colored and we are done.
\hfill\smallqed\medskip

\begin{claim}\label{clm:2-thread}
$G$ contains no pair of adjacent $2$-vertices.
\end{claim}
\claimproof
The proof is similar to the one of the previous claim. Suppose $u,v$ are two adjacent 2-vertices in $G$ and consider $G'=G-\{u,v\}$. Then $\Chisharp(G')\leq 4$ by minimality of $G$. It is easy to see that if the exact distance between two vertices of $V(G')$ in $G$ is 2, then this distance is preserved in $G'$. Thus an exact square coloring of $G'$ is valid in $G$ and since $u,v$ have degree at most 3 in $G\sharptwo$, they can be colored and we are done.
\hfill\smallqed\medskip

By Claim~\ref{clm:2-connected-K4minorfree}, $G$ is $2$-connected. Since $G$ is $K_4$-minor free, it has a nested ear decomposition $E_1,\ldots,E_k$ where $E_1$ is a cycle. Moreover, by Claim~\ref{clm:2-thread}, $G$ has no pair of adjacent $2$-vertices. Thus, there must be at least two ears in the decomposition. We consider a subsequence of $E_1,\ldots,E_k$ defined as follows: $E_{i_1}=E_1$; given $E_{i_j}$, the next ear $E_{i_{j+1}}$ is chosen (freely) among the ears nested on $E_{i_j}$ whose length of nest is as small as possible. Since our sequence is finite, this subsequence ends in an ear, say, $E_l$. By the minimality of the length of the nest of $E_l$, the vertices on its nest are all of degree~$2$. Since $l$ was the last element of the sequence, there is no ear nested on it, and all its internal vertices are of degree~2 (in $G$). Since $G$ has no pair of adjacent 2-vertices, $E_l$ and its nest each has at most one internal vertex. Thus $E_l$ together with its nest induce a cycle of length at most~4. However, if the length is~3 or~4, then there must be a degree 2-vertex on this cycle, contradicting Claim~\ref{clm:2-connected-K4minorfree}. Otherwise $E_j$ is nothing but a parallel edge, which is not possible since our graph is assumed to be simple.

For the tightness of the bound, in Proposition~\ref{prop:OmegasharpSP} we have presented a bipartite $K_4$-minor-free graph of girth~6 whose exact square is the disjoint union of two $K_4$'s. For the class of outerplanar graphs we have the example of Figure~\ref{fig:outerplanar}, given in~\cite{MO18} in the context of injective coloring of outerplanar graphs. 
\end{proof}

\section{Fullerene graphs}\label{sec:fullerenes}

We now turn our attention to a special class of cubic planar graphs, namely \emph{fullerene graphs}. They are the skeletons of cubic 3-dimensional convex polyhedra, each of whose faces is either a pentagon or a hexagon. As the skeleton of a convex polyhedron, a fullerene graph must be 3-connected and most authors consider this condition as part of the definition. However, we rather define a fullerene graph as a ``cubic plane graph each of whose faces is of size either 5 or 6". The following can then be proved as an exercise.

\begin{proposition}[Folklore]
\label{prop:fullerene_folklore}
Every fullerene graph has girth 5 and is 3-connected.
\end{proposition}

As any 3-connected graph admits a unique plane embedding, we may refer to a fullerene graph as a planar graph rather than a plane graph. Further results on their structure were proved, and we give here a few of them which we will use in the proofs of this section. Recall that a graph $G$ is \emph{cyclically $k$-edge-connected} if any edge-cut separating two cycles of $G$ has at least $k$ edges. Do\v{s}li\'c proved the following.

\begin{theorem}[\cite{D03}]\label{thm:Cyclic5-Connected}
Every fullerene graph is cyclically 5-edge-connected.
\end{theorem}

Using this theorem, one can easily classify all possible small cycles of fullerene graphs. More precisley, we have the following (perhaps folklore) fact; we refer to \cite[Lemma 4.1]{KKM19} for a proof.

\begin{lemma}\label{lem:small_cycle_fullerene}
Given a fullerene graph $G$, every non-facial cycle of $G$ is of length at least 9. Moreover, the only cycles of length 9, if any, are the cycles around a vertex incident with 5-faces only.
\end{lemma}

Note that $\Chisharp(G)\geq 3$ for any fullerene graph $G$.
We first characterize those fullerene graphs whose exact square is $3$-colorable. Then, we prove $\Chisharp(G)\leq 5$ for any fullerene graph $G$, thus proving Conjecture~\ref{conj} for this class of graphs.

\subsection{Characterizing exact square $3$-colorable fullerene graphs}\label{sec:fullerenes-3col}

We first need to define a special class of fullerene graphs, that we call \emph{drums}.

\begin{definition}
A drum is a fullerene graph with two specific 6-faces $F$ and $F'$, each of which is a neighbor with six 5-faces such that all these twelve 5-faces are distinct. A drum where $F$ and $F'$ are at facial distance $k+1$ is called a $k$-drum.
\end{definition}

For an example, the $1$-drum and $3$-drum are depicted in Figure~\ref{fig:example-drums}. In the literature, drums are known as a specific type of so-called \emph{nanotubes}, more precisely, following the terminology from~\cite{nanotubes,KS08}, they are exactly one of the five types of \emph{$(6,0)$-nanotubes}.

\begin{figure}
\subfloat[The 1-drum: the unique fullerene graph on $24$ vertices]{\label{subfig:24Fullerene}
	\begin{tikzpicture}[scale=0.55]	
	\tikzstyle{vertex}=[circle, draw, inner sep=0pt, minimum size=13pt,font=\footnotesize]
	
	\foreach \i in {1,...,6}{
	    \node[vertex] (c\i) at (60*\i-60:2) {$v_{\i}$};
	}
	\draw (c6) -- (c1) -- (c2) -- (c3) -- (c4) -- (c5) -- (c6);
	
	 \node[vertex] (v1) at (0:4) {$v'_{1}$};
	 \node[vertex] (v2) at (60:4) {$v'_{2}$};
	 \node[vertex] (v3) at (120:4) {$v'_{3}$};
	 \node[vertex] (v4) at (180:4) {$v'_{4}$};
	 \node[vertex] (v5) at (240:4) {$v'_{5}$};
	 \node[vertex] (v6) at (300:4) {$v'_{6}$};
	
	\foreach \i in {1,...,6}{
	    \draw (c\i) -- (v\i);
	}
	
	\node[vertex] (z1) at (30:5) {$u'_{1}$};
	\node[vertex] (z2) at (90:5) {$u'_{2}$};
	\node[vertex] (z3) at (150:5) {$u'_{3}$};
	\node[vertex] (z4) at (210:5) {$u'_{4}$};
	\node[vertex] (z5) at (270:5) {$u'_{5}$};
	\node[vertex] (z6) at (330:5) {$u'_{6}$};

	\foreach \i in {1,...,6}{
	    \draw (v\i) -- (z\i);
	}
	\draw (z6) -- (v1);
	\draw (z1) -- (v2);
	\draw (z2) -- (v3);
	\draw (z3) -- (v4);
	\draw (z4) -- (v5);
	\draw (z5) -- (v6);

	\foreach \i in {1,...,6}{
		\node[vertex] (u\i) at (-30+60*\i:7) {$u_{\i}$};
		\draw (u\i) -- (z\i);
	}
	
	\draw (u6) -- (u1) -- (u2) -- (u3) -- (u4) -- (u5) -- (u6);
	\end{tikzpicture}
}\hspace*{1cm}
\subfloat[The 3-drum $D_3$]{\label{subfig:$3$-drum}
	 \begin{tikzpicture}[scale=0.27]
	 \tikzstyle{vertex}=[circle, draw, inner sep=0pt, minimum size=6pt]
	
	 \foreach \i in {1,...,6}{
	 \node[vertex] (w\i) at (60*\i-60:2) {};
	 }
	 \draw (w1) -- (w2) -- (w3) -- (w4) -- (w5) -- (w6) -- (w1);
	
	 \foreach \i in {1,...,6}{
	 \node[vertex] (v\i) at (60*\i-60:4) {};
	 \draw (w\i) -- (v\i);
	 }
	
	 \foreach \i in {1,...,6}{
	 \node[vertex] (z\i) at (-30+60*\i:5) {};
	 \draw (v\i) -- (z\i);
	 }
	 \draw (z6) -- (v1);
	 \draw (z1) -- (v2);
	 \draw (z2) -- (v3);
	 \draw (z3) -- (v4);
	 \draw (z4) -- (v5);
	 \draw (z5) -- (v6);

	 \foreach \i in {1,...,6}{
	 \node[vertex] (u\i) at (-30+60*\i:7) {};
	 \draw (u\i) -- (z\i);
	 }

	 \foreach \i in {1,...,6}{
	 \node[vertex] (t\i) at (60*\i-60:9) {};
	 }
	
	 \draw (u1) -- (t2) -- (u2) -- (t3) -- (u3) -- (t4) -- (u4) -- (t5) --  (u5) -- (t6) -- (u6) -- (t1) -- (u1);
	
	 \foreach \i in {1,...,6}{
	 \node[vertex] (s\i) at (60*\i-60:11) {};
	 \draw (s\i) -- (t\i);
	 }
	
	 \foreach \i in {1,...,6}{
	 \node[vertex] (r\i) at (-30+60*\i:12) {};
	 \draw (r\i) -- (s\i);
	 }
	
	 \draw (r1) -- (s2);
	 \draw (r2) -- (s3);
	 \draw (r3) -- (s4);
	 \draw (r4) -- (s5);
	 \draw (r5) -- (s6);
	 \draw (r6) -- (s1);
	
	 \foreach \i in {1,...,6}{
	 \node[vertex] (q\i) at (-30+60*\i:15) {};
	 \draw (r\i) -- (q\i);
	 }
	
	 \draw (q1) -- (q2) -- (q3) -- (q4) -- (q5) -- (q6) -- (q1);
	 \end{tikzpicture}
}
\caption{Examples of drums.}
\label{fig:example-drums}
\end{figure}
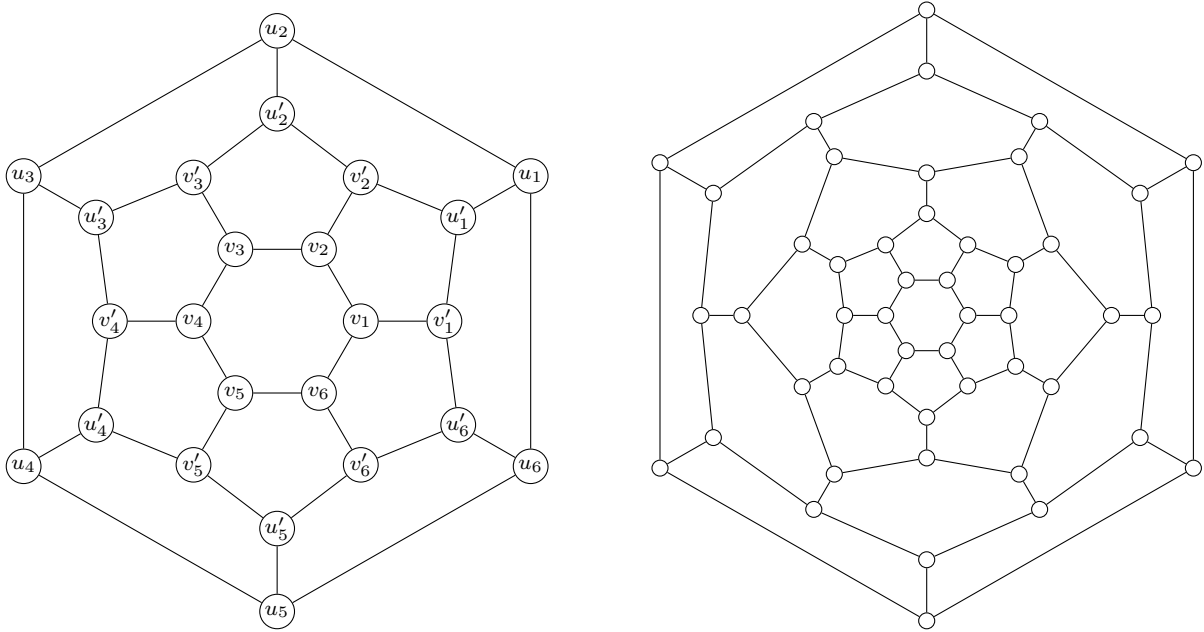

We will show next that there exists a unique $k$-drum up to isomorphism.

\begin{proposition}\label{prop:drumlayers}
  Given a $k$-drum $G$, all faces at distance $\ell \leq k$ from $F$ are of a same length.
\end{proposition}

\begin{proof}
We consider a planar embedding of $G$. Let $v_1,v_2, \ldots, v_6$ be the six vertices of $F$ in the cyclic order of vertices, and let $u_1,u_2, \ldots, u_6$ be the six vertices of $F'$ in the cyclic order of vertices, see Figure~\ref{subfig:24Fullerene} for the labeling of the 1-drum.

Observe that, since $G$ is cubic, each vertex $v_i$ has a unique neighbor in $G\setminus F$, we call it $v'_i$. Similarly, each neighbor of $u_i$ not in $F'$ is called $u'_i$. Observe that each pair of edges $v_iv'_i$ and $v_{i+1}v'_{i+1}$ (addition in indices here and in the rest of the proof are taken modulo $6$) is a pair of parallel edges of a $5$-cycle. 

To prove the claim of the proposition, for $\ell=1$, by the definition, all faces at distance 1 from $F$ are 5-faces. For $\ell=2$ we show that if one of the faces is a 5-face, then they are all 5-faces. So assume a face $f$ at distance 2 from $F$ is a 5-face. Then, as there are only twelve 5-faces, $f$ must be incident to $F'$ and since the 5-faces incident to $F'$ are distinct from those of $F$, the vertices of $f$ furthest away from $F$ form an edge of $F'$.  Thus, we may label the vertices of $f$, without loss of generality, $u_1u'_1v'_2u_2'u_2$. But then, the two faces incident to $u_1u'_1$ and $u_2u'_2$ are 5-faces next to $F'$; by continuing this process, we conclude that all faces at distance 2 from $F$ are the 5-faces incident to $F'$.

This completes the proof for $\ell=1,2$ with any value of $k$, which is exhaustive for $k\leq2$. For the remaining cases, we apply induction on $k$. 
 Assume that the claim is true for $k$ and all values of $\ell$, $\ell \leq k$. Consider a $(k+1)$-drum ($k\geq 2$). Thus, all faces at distance~2 from $F$ are 6-faces. In each of these 6-faces, one of the vertices is already labeled $v'_i$, noting that different faces correspond to different $v'_i$'s. 
Label $x_i$ the common neighbor of $v'_i$ and $v'_{i+1}$ (see Figure~\ref{fig:NeighbourhoodF}). Thus, $x_{i-1}v'_ix_{i}$ form part of a 6-face. On this face, label the neighbor of $x_i$ by $y_i$. Finally, label the common neighbor of $y_i$ and $y_{i-1}$ by $z_i$. Let $G'$ be the graph obtained from $G$ by deleting all the edges $y_iz_{i+1}$ and then contracting edges $z_iy_i$ and $y_ix_i$. We claim that $G'$ is a $(k-1)$-drum where faces at distance $\ell$ from $F$ in $G$ are at distance $\ell-1$ from $F$ in $G'$. This would complete the proof by induction. To see that $G'$ is a fullerene graph, observe that from the construction it is 3-regular. Each face of $G$ containing a path $z_iy_iz_{i+1}$ becomes a face of the same size on the path $x_{i}v'_{i+1}x_{i+1}$, and all other faces remain the same. Hence, we only have 5-faces and 6-faces in $G'$, so we are done.
\end{proof}

\begin{figure}[!ht]
\centering
	\begin{tikzpicture}[scale=0.4]
	\tikzstyle{vertex}=[circle, draw, inner sep=0pt, minimum size=13pt,font=\footnotesize]
	
	\node at (0,0) {$F$};
	
	\foreach \i in {1,...,6}{
	    \node[vertex] (v\i) at (60*\i-60:2) {$v_\i$};
	}
	\draw (v1) -- (v2) -- (v3) -- (v4) -- (v5) -- (v6) -- (v1);

	\foreach \i in {1,...,6}{
	    \node[vertex] (vv\i) at (60*\i-60:4) {$v'_\i$};
	    \draw (v\i) -- (vv\i);
	}
	
	\foreach \i in {1,...,6}{
	    \node[vertex] (x\i) at (-30+60*\i:5) {$x_\i$};
	    \draw (vv\i) -- (x\i);
	}
	\draw (x6) -- (vv1);
	\draw (x1) -- (vv2);
	\draw (x2) -- (vv3);
	\draw (x3) -- (vv4);
	\draw (x4) -- (vv5);
	\draw (x5) -- (vv6);

	\foreach \i in {1,...,6}{
	    \node[vertex] (y\i) at (-30+60*\i:7) {$y_\i$};
	    \draw (x\i) -- (y\i);
	}

	\foreach \i in {1,...,6}{
	    \node[vertex] (z\i) at (60*\i-60:9) {$z_\i$};
	}
	
	\draw (z1) -- (y1) -- (z2) -- (y2) -- (z3) -- (y3) -- (z4) -- (y4) --  (z5) -- (y5) -- (z6) -- (y6) -- (z1);
	
	\foreach \i in {1,...,6}{
	    \node[circle, draw, inner sep=0pt, minimum size=8pt] (u\i) at (60*\i-60:11) {};
	    \draw (u\i) -- (z\i);
	}
	
	\end{tikzpicture}
\caption{The neighborhood of face $F$ in a drum.}\label{fig:NeighbourhoodF}
\end{figure}
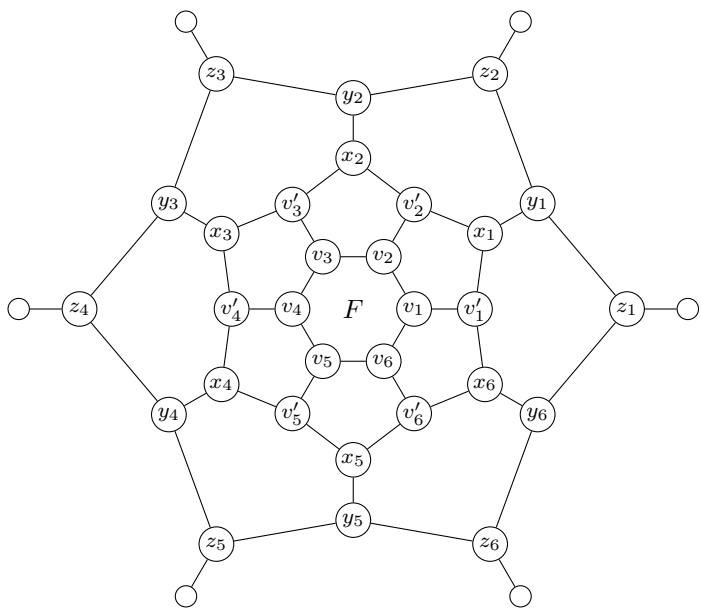

Our goal in this section is to characterize drums as the only fullerene graphs which are exact square 3-colorable. To this end, in the next lemma, we present two planar subcubic graphs that are not exact square 3-colorable, thus they cannot be induced subgraphs of an exact square 3-colorable graph.

\begin{lemma}\label{lem:no-C5With3C6}
Neither of the two graphs of Figure~\ref{fig:3c6-firstPrecoloring} admits an exact square 3-coloring.
\end{lemma}

\begin{proof}
We will repeatedly use the fact that in a proper 3-coloring of a $K_4^{-}$ (that is, the complete graph on four vertices minus one edge), the nonadjacent vertices must receive a same color. Applying this observation to the exact square of the graphs of Figure~\ref{fig:3c6-firstPrecoloring}, we conclude that in a hypothetical 3-coloring of each of them, vertices $v_5$, $t_2$, $u_3$ and $w_4$ must receive a same color (red or diamond shaped in the figure). This is already a contradiction in the graph of Figure~\ref{subfig:3surronding_C6}, as vertices $w_4$ and $v_5$ are at distance~2. 

To complete the proof for the graph of Figure~\ref{subfig:4surronding_C6}, observe that for the same reason, vertices $u_1$, $t_3$ and $v_4$ must also get a same color, and this color must be distinct from red because $u_3$ and $v_4$ are at distance~2. We suppose this color to be green (pentagon shaped in the figure). Furthermore, vertices $v_2$ and $t_6$ are colored by the third color (blue or hexagon shaped), because each of them sees both other colors at distance~2. Moreover, since $w_5$ must receive the same color as $t_6$, we conclude that $w_5$ must be colored blue.

Now, we know that $v_1$ sees both blue and green at distance~2, and thus it must be colored red. Repeating the $K_4^{-}$ argument, we conclude that $t_4$ and then $u_4$ must be colored red as well. This is a contradiction, as $u_4$ is at distance~2 from the vertex $v_5$, colored red.
\end{proof}

\begin{figure}[h!]
\centering
\subfloat[]{\label{subfig:3surronding_C6}
	\begin{tikzpicture}[scale=0.6]
	\tikzstyle{vertex}=[circle,draw,thick,fill=white,inner sep=1pt]
	\tikzstyle{vertexr}=[diamond,draw,thick,inner sep=0.4pt]

	\foreach \i in {1,...,4}{
	\node[vertex] (v\i) at (30+60*\i:2) {$v_{\i}$};
	}

	\node[vertexr] (v5) at (0:2) {$v_5$};
	
	\draw (v1) -- (v2) -- (v3) -- (v4) -- (v5) -- (v1);
	
	\foreach \i in {1,2,4}{
	\node[vertex] (u\i) at (30+60*\i:4) {$u_{\i}$};
	}

	\node[vertexr] (u3) at (210:4) {$u_3$};

	\foreach \i in {1,3}{
	\node[vertex] (t\i) at (80+30*\i:5.2) {$t_{\i}$};
	}
	\node[vertex] (t5) at (230:5.4) {$t_{5}$};
	
	\foreach \i in {4,6}{
	\node[vertex] (t\i) at (72+30*\i:5.2) {$t_{\i}$};
	}
    \node[vertexr] (t2) at (132:5.2) {$t_2$};
	
	\foreach \i in {1,...,4}{
	\draw (v\i) -- (u\i);
	}
	
	\draw (u1) -- (t1) -- (t2) -- (u2) -- (t3) -- (t4) -- (u3) -- (t5) -- (t6) -- (u4);

	\node[vertex] (t3) at (170:5.2) {$t_3$};
	\node[vertex] (u1) at (90:4) {$u_1$};
	\node[vertex] (v4) at (270:2) {$v_4$};
	
	\node[vertex] (t6) at (252:5.2) {$t_6$};
	\node[vertex] (v2) at (150:2) {$v_2$};

	\node[vertexr, right of=u4,node distance=2.5cm] (w4)  {$w_4$};
	
	\node[vertex, right of=v4,node distance=2.5cm] (w5) {$w_5$};
	
	\draw (u4) -- (w4) --(w5)-- (v5);
	\end{tikzpicture}
}\hspace*{0.8cm}
\subfloat[]{\label{subfig:4surronding_C6}
\begin{tikzpicture}[scale=0.55]
	\tikzstyle{vertex}=[circle,draw,thick,fill=white,inner sep=1pt]
	\tikzstyle{vertexb}=[regular polygon,regular polygon sides=5,draw,thick,inner sep=0pt]
	\tikzstyle{vertexr}=[diamond,draw,thick,inner sep=0.4pt]
	\tikzstyle{vertexg}=[regular polygon,regular polygon sides=6,draw,thick,inner sep=0pt]

	\foreach \i in {1,3}{
	\node[vertex] (v\i) at (30+60*\i:2) {$v_{\i}$};
	}
	\node[vertexb] (v2) at (150:2) {$v_2$};
	\node[vertexg] (v4) at (270:2) {$v_4$};
	\node[vertexr] (v5) at (0:2) {$v_5$};

	\draw (v1) -- (v2) -- (v3) -- (v4) -- (v5) -- (v1);
	
	\foreach \i in {2,4}{
	\node[vertex] (u\i) at (30+60*\i:4) {$u_{\i}$};
	}
	\node[vertexg] (u1) at (90:4) {$u_1$};
	\node[vertexr] (u3) at (210:4) {$u_3$};

	\foreach \i in {1}{
	\node[vertex] (t\i) at (80+30*\i:5.2) {$t_{\i}$};
	}
	\node[vertex] (t5) at (230:5.4) {$t_{5}$};
	\node[vertexg] (t3) at (170:5.2) {$t_3$};
	
	\foreach \i in {4}{
	\node[vertex] (t\i) at (72+30*\i:5.2) {$t_{\i}$};
	}
	\node[vertexr] (t2) at (132:5.2) {$t_2$};
	\node[vertexb] (t6) at (252:5.2) {$t_6$};
	
	\foreach \i in {1,...,4}{
	\draw (v\i) -- (u\i);
	}
	
	\draw (u1) -- (t1) -- (t2) -- (u2) -- (t3) -- (t4) -- (u3) -- (t5) -- (t6) -- (u4);
	
	\node[vertexr, below of=v5,node distance=3.4cm] (w4)  {$w_4$};
	
	\node[vertexb, right of=v4,node distance=2.2cm,inner sep=-0.9pt] (w5) {$w_5$};
	\node[vertex, right of=u4,node distance=2.2cm, minimum size=15pt] (s) {};
	
	\draw (u4) -- (w4) -- (s) --(w5)-- (v5);
	
	\end{tikzpicture}
}
\caption{The forced precolorings from the proof of Lemma~\ref{lem:no-C5With3C6}.}
\label{fig:3c6-firstPrecoloring}
\end{figure}
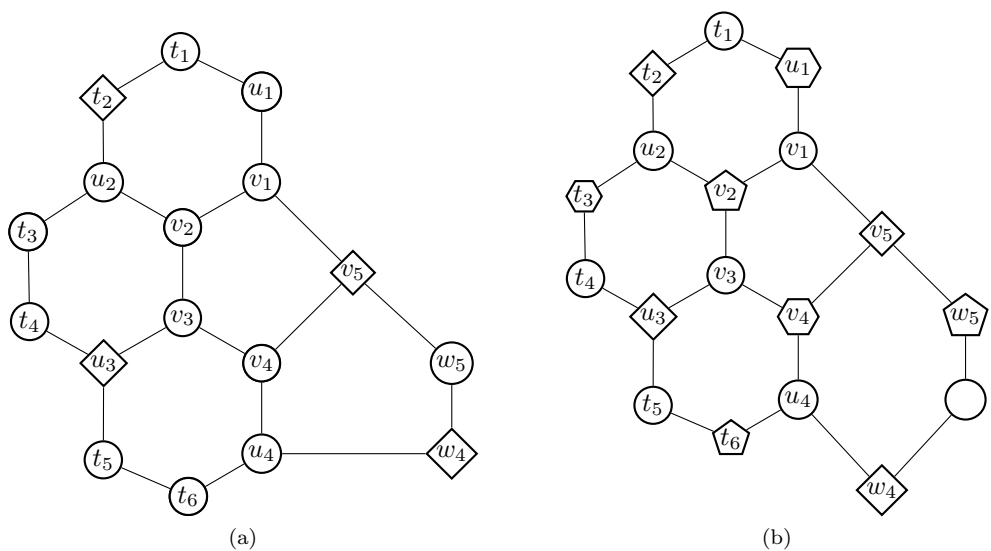

Similarly, we show in the next lemma, that the graph of Figure~\ref{fig:3C5-1}, while admitting an exact square 3-coloring, has limits on its possible 3-colorings.

\begin{lemma}\label{lem:C5SurrondedBy3C5}
In an exact square 3-coloring of the graph $G$ of Figure~\ref{fig:3C5-1}, vertices $x$ and $y$ must receive distinct colors.
\end{lemma}

\begin{proof}
By contradiction, suppose $\phi$ is an exact square 3-coloring of $G$ with $\phi(x)=\phi(y)=1$. Then, without loss of generality, we can assume that $\phi(v_1)=2$ and $\phi(v_3)=3$. Therefore, we have $\phi(v_5)=2$. But then, $\phi(u_2)=2$ because $u_2$ sees $v_1$ and $y$ at distance~2. Similarly, $\phi(u_4)=2$. This is a contradiction since $u_2$ and $u_4$ see each other at distance~2.
\end{proof}

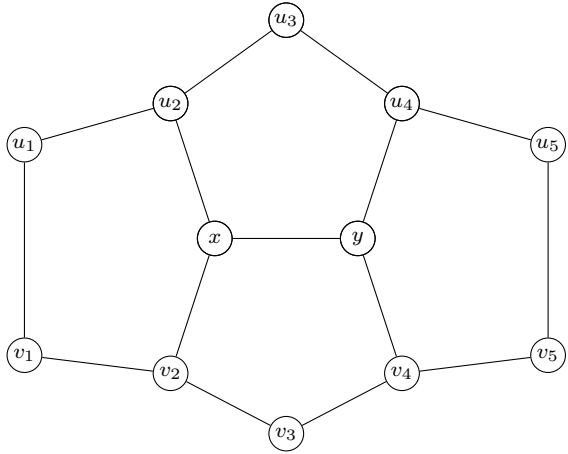
\begin{figure}[!ht]
\centering
	\begin{tikzpicture}[scale=0.8]
	\tikzstyle{vertex}=[circle, draw, inner sep=0pt, minimum size=13pt,font=\footnotesize]
	
	\foreach \i in {1,...,5}{
	    \node[vertex] (u\i) at (72*\i+18:2) {};
	}
	\draw (u1) -- (u2) -- (u3) -- (u4) -- (u5) -- (u1);
	
	\node[vertex] (u1) at (90:2) {$u_3$};
	\node[vertex] (u2) at (162:2) {$u_2$};
	\node[vertex] (u3) at (234:2) {$x$};
	\node[vertex] (u4) at (306:2) {$y$};
	\node[vertex] (u5) at (18:2) {$u_4$};
	
	\node[vertex] (v3) at (90:-4.8541) {$v_3$};
	
	\node[vertex] (v2) at (-1.9021,-3.854) {$v_2$};
	
	\node[vertex] (v4) at (1.9021,-3.854) {$v_4$};
	
	\draw (u3) -- (v2) -- (v3) -- (v4) -- (u4);
	
	\node[vertex] (u1) at (-4.3042,-0.064) {$u_1$};
	
	\node[vertex] (v1) at (-4.3042,-3.5539) {$v_1$};
	
	\draw (u2) -- (u1) -- (v1) -- (v2);
	
	\node[vertex] (u5b) at (4.3042,-0.064) {$u_5$};
	
	\node[vertex] (v5) at (4.3042,-3.5539) {$v_5$};
	
	\draw (u5) -- (u5b) -- (v5) -- (v4);
	
	\end{tikzpicture}
\caption{A 5-cycle surrounded by three consecutive 5-cycles.} \label{fig:3C5-1}
\end{figure}

\begin{lemma}
\label{lem:no-3C5-1}
If $G$ is a fullerene graph which is not the 1-drum and contains the graph of Figure~\ref{fig:3C5-1} as a subgraph, then $\Chisharp(G) \geq 4$.
\end{lemma}

\begin{proof}
Let $G$ be an exact square 3-colorable fullerene graph which contains the graph of Figure~\ref{fig:3C5-1} as a subgraph. Moreover, let $\phi$ be its exact square 3-coloring. By Lemma~\ref{lem:C5SurrondedBy3C5} and without loss of generality, we may assume that $\phi(x)=1$ and $\phi(y)=2$. Hence, $\phi(u_3)=\phi(v_3)=3$ and $\{\phi(u_2), \phi(v_2)\}=\{1, 3\}$  and  $\{\phi(u_4), \phi(v_4)\}=\{2, 3\}$. By the symmetry along the edge $xy$, we may assume that $\phi(u_2)=3$, which then implies $\phi(v_2)=1$, $\phi(u_4)=2$ and $\phi(v_4)=3$. Therefore, we have $\phi(u_1)=\phi(v_1)=2$ and $\phi(u_5)=\phi(v_5)=1$. 

Next, noting that $G$ is a 3-regular graph, we consider the remaining neighbors of degree~2 vertices of this subgraph. Let $a,b,c,d,e,f$ be, respectively, the neighbors of $u_1,u_3, u_5, v_5,v_3,v_1$. The coloring extends uniquely to these six vertices as follows: $\phi(a)=\phi(b)=1$, $\phi(c)=\phi(f)=3$, and $\phi(d)=\phi(e)=2$.
Since $G$ is planar and cubic, vertices $a$ and $b$ are lying on the same face. Moreover, since $G$ can have only $5$-faces and $6$-faces, we conclude that $a$ and $b$ (which are colored with the same color) must be adjacent. Similarly, we conclude that $e$ and $d$ must be adjacent.
On the other hand, $b$ and $c$ cannot be adjacent, as otherwise $b$ and $u_5$ would be at distance~2 while both having the same color. Hence, $b$ and $c$ have a common neighbor, say $b'$, and thus we get a 6-face $u_5u_4u_3bb'cu_5$, which we name $F$. Similarly, vertices $e$ and $f$ have a common neighbor, say $e'$, thus we get a 6-face $v_1v_2v_3ee'fv_1$ and we name it $F'$.
Now, vertices $a$ and $f$ must lie on the same face, which can be either a 5-face or a 6-face. Observe also that the third neighbor of $a$, say $a'$, distinct from $u_1$ and $b$ must be colored~3. Therefore, since $f$ is colored~3, vertex $a'$ cannot be at distance~2 from $f$ and thus we conclude that $a$ and $f$ are lying on a 5-face $au_1v_1fa'a$. Symmetrically, we get that vertices $c$ and $d$ are lying on a 5-face as well: $dv_5u_5cc'd$ (where $c'$ is the common neighbor of $c$ and $d$).

In summary, starting from the graph of Figure~\ref{fig:3C5-1}, with $xy$ being the central edge, we concluded that the neighboring structure is forced. But we now have other isomorphic copies of the graph of Figure~\ref{fig:3C5-1} inside $G$, for example one centered around $u_1u_2$ and another one centered around $v_4v_5$. Thus, the same local neighborhood structures should exist around these edges. The 6-faces in these structures are already given ($F$ and $F'$). Thus, we conclude that $e',e,d,c'$ are lying on a 5-face. Similarly, $a',a,b,b'$ are lying on a 5-face as well. This forces a graph where all but two vertices have degree~3, the other two vertices being of degree~2. To complete this to a 3-connected cubic graph, we then join these two vertices and obtain the 1-drum.
\end{proof}

 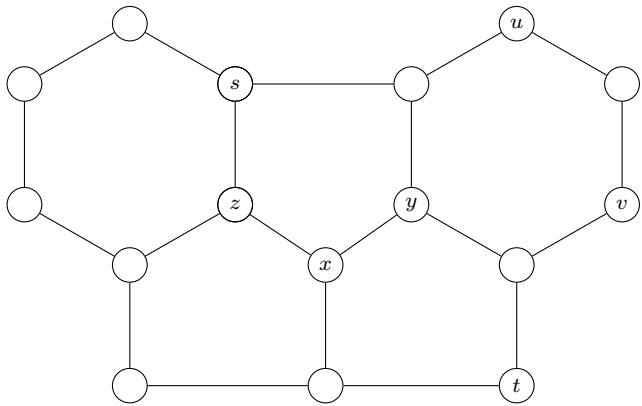
\begin{figure}[!ht]
\centering
	\begin{tikzpicture}[scale=0.8]
	\tikzstyle{vertex}=[circle, draw, inner sep=0pt, minimum size=13pt,font=\footnotesize]

	\foreach \i in {1,...,6}{
	    \node[vertex] (u\i) at (60*\i+30:2) {};
	}
	\draw (u1) -- (u2) -- (u3) -- (u4) -- (u5) -- (u6) -- (u1);
	
	\node[vertex] (u5) at (330:2) {$z$};
	\node[vertex] (u6) at (30:2) {$s$};
	\node[vertex] (x) at (3.218,-2) {$x$};
	\node[vertex] (y) at (4.6282,-1) {$y$};
	\node[vertex] (s') at (4.6282,1) {};
	
	\draw (u5) -- (x) -- (y) -- (s') -- (u6);
	
	\node[vertex] (x') at (6.3602,-2) {};
	\node[vertex] (u1') at (6.3602,2) {$u$};
	
	\node[vertex] (s'') at (8.0922,1) {};
	\node[vertex] (y') at (8.0922,-1) {$v$};
	
	\draw (s') -- (u1') -- (s'') -- (y') -- (x') -- (y);
	
	\node[vertex] (a) at (0,-4) {};
	\node[vertex] (b) at (3.218,-4) {};
	\node[vertex] (c) at (6.3602,-4) {$t$};
	
	\draw (u4) -- (a) -- (b) -- (x);
	\draw (b) -- (c) -- (x');
	
	\end{tikzpicture}
\caption{The graph of Lemma~\ref{lem:no-3C5-2}.}
 \label{fig:3C5-2}
\end{figure}

\begin{lemma}
\label{lem:no-3C5-2}
The graph of Figure~\ref{fig:3C5-2} does not admit an exact square 3-coloring.
\end{lemma}

\begin{proof}
To prove the lemma, we claim that in any possible exact square 3-coloring of this graph, vertices $x$ and $y$ must receive a same color. Considering the symmetry of edges $xy$ and $xz$, the same argument then would apply to $x$ and $z$. This would lead to a contradiction, since $y$ and $z$ cannot be colored the same.

To prove the claim, assume $\phi$ is an exact square 3-coloring of the graph of Figure~\ref{fig:3C5-2}, and without loss of generality, assume that $\phi(y)=1$, $\phi(u)=2$ and $\phi(v)=3$. Then, $\phi(s)=3$ and $\phi(t)=2$. This in turn implies $\phi(x)=1$, which is the color of $y$.
\end{proof}

\begin{figure}[!ht]
\centering
	\begin{tikzpicture}[scale=0.8]
	\tikzstyle{vertex}=[circle, draw, inner sep=0pt, minimum size=13pt,font=\footnotesize]
	\foreach \i in {1,...,6}{
	    \node[vertex] (u\i) at (60*\i+30:2) {};
	}
	\draw (u1) -- (u2) -- (u3) -- (u4) -- (u5) -- (u6) -- (u1);
	
	\node[vertex] (u1) at (90:2) {$s_1$};
	\node[vertex] (u2) at (150:2) {$t_1$};
	\node[vertex] (u3) at (210:2) {$z_1$};
	\node[vertex] (u4) at (270:2) {$y_2$};
	\node[vertex] (u5) at (330:2) {$z_2$};
	\node[vertex] (u6) at (30:2) {$t_2$};
	
	\node[vertex] (s2) at (3.464,2) {$s_2$};
	\node[vertex] (t3) at (5.196,1) {$t_3$};
	\node[vertex] (z3) at (5.196,-1) {$z_3$};
	\node[vertex] (y3) at (3.464,-2) {$y_3$};
	
	\draw (u6) -- (s2) -- (t3) -- (z3) -- (y3) -- (u5);
	
	\node[vertex] (x2) at (0,-4.5) {$x_2$};
	\node[vertex] (x3) at (3.464,-4.5) {$x_3$};
	
	\node[vertex] (y1) at (-3.464,-2) {$y_1$};
	\node[vertex] (x1) at (-3.464,-4.5) {$x_1$};
	
	\node[vertex] (y4) at (6.928,-2) {$y_4$};
	\node[vertex] (x4) at (6.928,-4.5) {$x_4$};
	
	\draw (u3) -- (y1) -- (x1) -- (x2) -- (u4);
	\draw (x2) -- (x3) -- (x4) -- (y4) -- (z3);
	\draw (x3) -- (y3);
	
	\node[vertex] (u) at (6.928,2) {$u$};
	\node[vertex] (v) at (8.66,0) {$v$};
	
	\draw (t3) -- (u) -- (v) -- (y4);
	\end{tikzpicture}
\caption{The graph of Lemma~\ref{lem:C5seq}.}
 \label{fig:C5seq}
\end{figure}
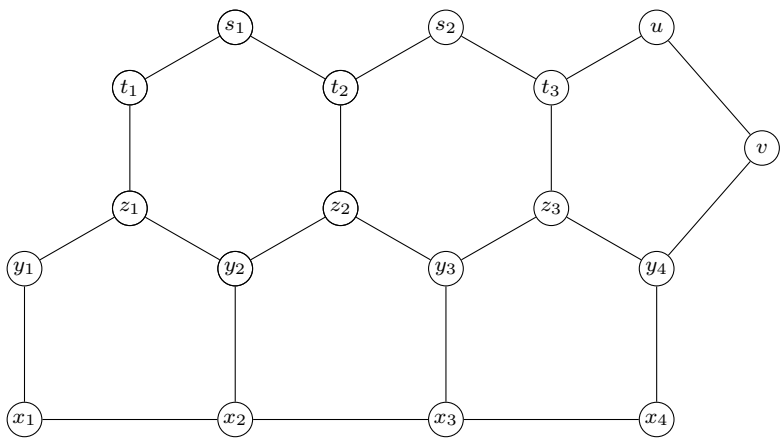

\begin{lemma}
\label{lem:C5seq}
The graph of Figure~\ref{fig:C5seq} does not admit an exact square 3-coloring.
\end{lemma}

\begin{proof}
By contradiction, suppose $\phi$ is a 3-coloring of the exact-square of this graph.
We first claim that $\phi(x_2)=\phi(y_2)$. If not, then we may assume $\phi(x_2)=1$ and $\phi(y_2)=2$, then $\phi(y_1)=\phi(y_3)=3$, which in turn implies that $\phi(t_1)=\phi(t_2)=1$ but $t_1$ and $t_2$ are at distance 2. 
Note that this proof is based solely on the three faces around $y_2$, so similarly, if a vertex has two 5-faces and one 6-face around it, then its color must be the same as the color of its neighbor on the 5-faces. Applying this to our graph we have: $\phi(x_3)=\phi(y_3)$, $\phi(z_3)=\phi(y_4)$.

Suppose $\phi(x_2)=\phi(y_2)=2$, then as $y_3$ is at distance 2 from $y_2$, and by the symmetry of other colors, we have $\phi(x_3)=\phi(y_3)=3$, then $\phi(z_2)=1$ and since $z_3$ is at distance 2 from both $z_2$ and $x_3$ we have $\phi(z_3)=\phi(y_4)=2$. This in turn implies that $\phi(t_3)=1$, but then $t_2$ sees all three colors at distance 2, that is $y_2$ for color 2, $y_3$ for color 3 and $t_3$ for color 1.
\end{proof}

We can now state the main theorem of this section.

\begin{theorem}\label{thm:fullerene-3-col}
A fullerene graph is exact square $3$-colorable if and only if it is a $k$-drum, for a positive integer $k$.
\end{theorem}

\begin{proof}
First, we show that every drum is exact square 3-colorable. Given a $k$-drum, consider its 6-face $F$ surrounded by six 5-faces. Color the vertices of $F$ with $1,2,3,1,2,3$ in a clockwise orientation of $F$. Observe that, by Proposition~\ref{prop:drumlayers}, for $2 \leq l \leq k$, there are six faces of the same length at distance $l$ from $F$. If we consider the subgraph $G_l$ induced by faces at distance at most $l$ from $F$, its outer face is a 12-cycle. The 3-coloring of $F$ is then uniquely extended to a 3-coloring of $G_l$, where the 12 vertices of the outer face are colored consecutively $1,2,3$ in the counterclockwise orientation with respect to $F$. 

Similarly, if we start from the face $F'$ of the $k$-drum, and color the vertices of $F'$ with $1,2,3$ in the counterclockwise orientation (of $F'$), then the coloring extends uniquely to any subgraph obtained by faces at distance at most $k-l$ from $F'$. In such a coloring, vertices of the outer 12-face are colored consecutively $1,2,3$ with clockwise orientation with respect to $F'$. As this orientation of the outer 12-face matches its counterclockwise orientation with respect to $F$, we can choose a proper rotation of colors on $F'$, so that we can merge the colorings of both parts of the drum. Notice that, except the vertices of the 12-cycle, all the vertices of one part of the drum are at distance at least~3 from vertices of the other part.

It remains to show that if a fullerene graph admits an exact square 3-coloring, then it must be a drum. Let $G$ be an exact square 3-colorable fullerene graph. By Lemma~\ref{lem:no-C5With3C6}, a 5-face cannot have three consecutive 6-faces. By Lemmas~\ref{lem:no-3C5-1} and~\ref{lem:no-3C5-2}, a vertex cannot be incident to three 5-faces. This leaves us with one possibility: each 5-face $C$ is neighbor with two other 5-faces through two non-adjacent edges of $C$.
Labelling the vertices of one of the 5-faces, in the cyclic order,  $x_2,y_2,z_2,y_3,x_3$, we may assume that each of the edges $x_2y_2$ and $x_3y_3$ is incident to another 5-face, thus so far we have the subgraph of Figure~\ref{fig:C5seq} induced by vertices $x_i$'s, $y_i$'s and $z_i$'s, for $1\leq i \leq 3$.
Using the labeling of this figure, the claim of Lemma~\ref{lem:C5seq} implies that the second 5-face neighbor of the face $x_3y_3z_3y_4x_4$ cannot be on the edge $z_3y_4$, thus it must be on $y_4x_4$. Completing this sequence of 5-faces, we conclude that vertices $x_i$ form a face of $G$ as they are already of full degree. This face then can only be a 6-face which is the face $F$ of the drum. The face $F'$ of drum is found the same way by considering the remaining 5-faces. 
\end{proof}

\subsection{Fullerene graphs are exact square 5-colorable}\label{sec:fullerenes-5col}

To prove the main result of this section, we first give two lemmas on proper coloring of some graphs. Specifically, we show that the precoloring extension of two graphs (where vertices take colors from lists of given sizes) can always be done when the lists are subsets of $\{1,2,3,4,5\}$.

\begin{lemma}\label{lemma:6-wheel}
Let $abcdef$ be a 6-cycle and $x$ be a vertex such that $N(x)=\{a,b,c,d,e,f\}$. Suppose these vertices have lists of available colors from the set $\{1,2,3,4,5\}$ satisfying the following: $|L(a)|\geq 2$, $|L(b)|\geq 2$, $|L(c)|\geq 2$, $|L(d)|\geq 2$, $|L(e)|\geq 2$, $|L(f)|\geq 3$ and $|L(x)|=5$ (see Figure~\ref{fig:6-wheel}). Then there exists a proper $L$-coloring of this graph. 
\end{lemma}
\begin{proof}

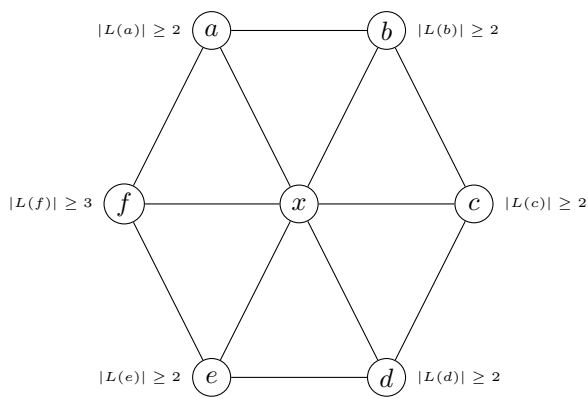
\begin{figure}[!ht]
\centering
\begin{tikzpicture}[auto,scale=2.3]
\tikzstyle{vertex}=[draw, circle, inner sep=0.55mm]
\node[minimum size=0.5cm] (a) at (0,0) [vertex] {$a$};
\node [left=0.5bp of a] {\tiny $|L(a)|\geq 2$};
\node[minimum size=0.5cm] (b) at (1,0) [vertex] {$b$};
\node [right=0.5bp of b] {\tiny $|L(b)|\geq 2$};
\node[minimum size=0.5cm] (c) at  (1.5,-1) [vertex] {$c$};
\node [right=0.5bp of c] {\tiny $|L(c)|\geq 2$};
\node[minimum size=0.5cm] (d) at (1,-2) [vertex] {$d$};
\node [right=0.5bp of d] {\tiny $|L(d)|\geq 2$};
\node[minimum size=0.5cm] (e) at (0,-2) [vertex] {$e$};
\node [left=0.5bp of e] {\tiny $|L(e)|\geq 2$};
\node[minimum size=0.5cm] (f) at (-.5,-1) [vertex] {$f$};
\node [left=0.5bp of f] {\tiny $|L(f)|\geq 3$};
\node[minimum size=0.5cm] (x) at  (.5,-1) [vertex] {$x$};
\node [below right=0.5bp and 0.35bp of x] {}; 

\draw (a)--(b)--(c)--(d)--(e)--(f)--(a);
\draw (a) to (x);
\draw (b) to (x);
\draw (c) to (x);
\draw (d) to (x);
\draw (e) to (x);
\draw (f) to (x);
\end{tikzpicture}
\caption{Graph of Lemma~\ref{lemma:6-wheel}, where $L(x)=\{1,2,3,4,5\}$.}
\label{fig:6-wheel}
\end{figure}

First we consider the case that $L(a)\cap L(c)\neq\emptyset$ and assume, without loss of generality, that $1\in L(a)\cap L(c)$. We color $\phi(a)=\phi(c)=1$. We then color vertices $b,d,e,f$ in this order by observing that at each step there is an available color for the current vertex. Now, if there is an available color from $L(x)$ left for $x$, we are done. If not, then $b,d,e,f$ have each been colored with a distinct color from $\{2,3,4,5\}$, say $\phi(b)=2, \phi(d)=3, \phi(e)=4, \phi(f)=5$. We conclude that $L(b)=\{1,2\}$, $L(f)=\{1,4,5\}$, $L(e)\subset \{3,4,5\}$ and $L(d) \subset\{1,3,4\}$ as otherwise one of these vertices could be recolored in order to gain a free color for vertex $x$. If $1\in L(d)$, then we consider another coloring $\psi(b)=\psi(d)=\psi(f)=1$. This coloring then extends to a coloring of $a,c$ and $e$, after which only four colors are used and thus we have a color left for $x$.  Hence we may assume $L(d)=\{3,4\}$. Let $\alpha\neq 4$ be a color in $L(e)$, then the coloring $\phi(a)=\phi(c)=1$, $\phi(b)=2$, $\phi(d)=\phi(f)=4$, $\phi(e)=\alpha$ uses at most four colors on neighbors of $x$ and thus there is a color available at $x$. 

As the pair $e,c$ is symmetric to the pair $a,c$, for the remaining cases we may assume that $L(a)\cap L(c)=L(c)\cap L(e)=\emptyset$. Since the colors are taken from the set $\{1,2,3,4,5\}$, we conclude that $L(a)\cap L(e)\neq \emptyset$, and without loss of generality, we assume the coloring $\phi(a)=\phi(e)=1$. If one of the two lists $L(b)$ or $L(d)$ does not contain color 1, then one could color the path $bcd$, and afterwards color vertex $x$ and then finish by coloring $f$. Hence $1\in L(b)\cap L(d)$ and then we consider another coloring $\psi(b)=\psi(d)=1$, which extends to the path $afe$ (by first coloring $a$ and $e$). We then color vertex $x$. Finally, since we have already established that $1\not\in L(c)$, we have a color left for vertex $c$ and we are done.
\end{proof}

\begin{lemma}\label{lemma:list_col_triangulation}
The graph of Figure~\ref{fig:list_col_triangulation} with the given lower bounds on the sizes of lists of colors from the set $\{1,2,3,4,5\}$ is $L$-colorable.
\end{lemma}

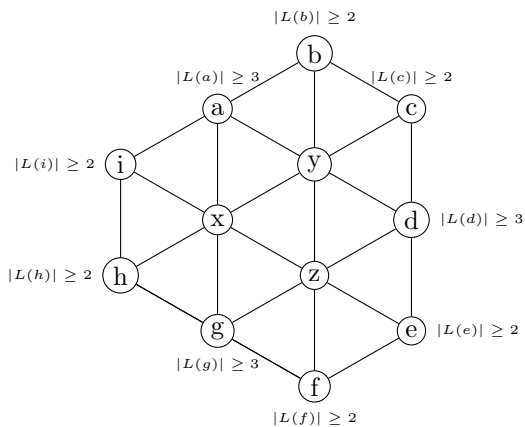
\begin{figure}[!ht]
\centering

\begin{tikzpicture}[scale=0.85]
\tikzstyle{vertex}=[draw, circle, minimum size = 4pt, inner sep=0.55mm]

  \node (b) at (3*cos{60}+1,5*sin{60}) [vertex] {b};
  \node [above=0.5bp of b] {\tiny $|L(b)|\geq 2$};
  \node (a) at (1,4*sin{60}) [vertex] {a};
  \node [above=0.2bp of a] {\tiny $|L(a)|\geq 3$};
  \node (c) at (4,4*sin{60}) [vertex] {c};
  \node [above=0.2bp of c] {\tiny $|L(c)|\geq 2$};
  \node (i) at (-1*cos{60},3*sin{60}) [vertex] {i};
  \node [left=0.2bp of i] {\tiny $|L(i)|\geq 2$};
  \node (y) at (3*cos{60}+1,3*sin{60}) [vertex] {y};
  \node (x) at (1,2*sin{60}) [vertex] {x};
  \node (d) at (4,2*sin{60}) [vertex] {d};
  \node [right=0.2bp of d] {\tiny $|L(d)|\geq 3$};
  \node (h) at (-1*cos{60},sin{60}) [vertex] {h};
  \node [left=0.2bp of h] {\tiny $|L(h)|\geq 2$};
  \node (z) at (3*cos{60}+1,sin{60}) [vertex] {z};
  \node (g) at (1,0) [vertex] {g};
  \node [below=0.2bp of g] {\tiny $|L(g)|\geq 3$};
  \node (e) at (4,0) [vertex] {e};
  \node [right=0.2bp of e] {\tiny $|L(e)|\geq 2$};
  \node (f) at (3*cos{60}+1,-1*sin{60}) [vertex] {f};
  \node [below=0.2bp of f] {\tiny $|L(f)|\geq 2$};
  
  \draw (b)--(a)--(i)--(h)--(g)--(f)--(e)--(d)--(c)--(b);
  \draw (b)--(y)--(a) (x)--(y)--(z) (x)--(z)--(g)
  		(i)--(x)--(h) (x)--(g)--(h) (g)--(f)--(z)
  		(z)--(y)--(d) (e)--(z)--(d) (x)--(a) (y)--(c);
\end{tikzpicture}
\caption{Graph of Lemma~\ref{lemma:list_col_triangulation}, where $L(x)=L(y)=L(z)=\{1,2,3,4,5\}$.}
\label{fig:list_col_triangulation}
\end{figure}

\begin{proof}
Let $G$ be the graph from the statement of the lemma. We distinguish three cases:
\begin{enumerate}
\item Suppose $L(f)\cap L(d)\neq \emptyset$ and let $\alpha\in L(f)\cap L(d)$. Then we give the following partial coloring of $G$: we color both $f$ and $d$ with color $\alpha$ and then greedily color $e,c,b$, in this order (at each step, there is at least one available color with respect to $L$). Now the remaining uncolored vertices $a,y,z,g,h,i,x$ form the configuration of Lemma~\ref{lemma:6-wheel}, so we are done.

\item We have $L(f)\cap L(d)=\emptyset$. Now, suppose $L(f)\cap L(e)\neq\emptyset$. Then we color $e$ with some color $\alpha\in L(f)\cap L(e)$ and thus we still have $|L(d)|\geq 3$. We then color greedily $f,g,h,i$, in this order. The remaining uncolored vertices $z,x,a,b,c,d,y$ form the configuration of Lemma~\ref{lemma:6-wheel}, so we are done.

\item By the previous items we have that $L(f)\cap L(d)=\emptyset$ and $L(f)\cap L(e)=\emptyset$. Moreover, the same holds for each pair of vertices isomorphic to $f,d$ or to $f,e$. As the lists are taken from the set $\{1,2,3,4,5\}$, without loss of generality, we can assume that $L(f)=\{1,2\}$, $\{3,4\}\subseteq L(e)$ and $L(d)=\{3,4,5\}$. Moreover, by symmetry we have $L(g)\cap L(e)=\emptyset$ and thus $L(g)=\{1,2,5\}$. Therefore $|L(g)\cap L(d)|=1$.
Now note that vertex $a$ is symmetric to vertex $d$ and to vertex $g$ and thus must satisfy $|L(a)\cap L(d)|=1$ and $|L(a)\cap L(g)|=1$, which is impossible since the colors are taken from a set of five elements.\qedhere
\end{enumerate}
\end{proof}

Let $\mathcal C$ be the class of induced subgraphs of fullerene graphs. In order to show that every fullerene graph is exact square 5-colorable, we will prove a stronger statement: that every graph in $\mathcal C$ is exact square 5-colorable.

The main result of this section is the following theorem.

\begin{theorem}
Every graph $G\in \mathcal{C}$ is exact square $5$-colorable.
\end{theorem}
\begin{proof}

Let $G$ be a graph of $\mathcal C$ that is a minimum counterexample to our claim, that is, it is of smallest order among those that are not exact square $5$-colorable. Let $H$ be a fullerene graph that contains $G$ as an induced subgraph. In the remainder of the proof, $G$ will be a regarded as a plane graph whose embedding is induced by the unique embedding of $H$. By Lemma~\ref{lem:small_cycle_fullerene}, we know that every 5-cycle or 6-cycle of $G$ is a face of both $G$ and $H$. Furthermore, $G$ has no other cycle of length less than~9, and all cycles of length~9 are obtained from the symmetric differences of three 5-faces sharing a common vertex. Further properties of $G$ are as follows.

\begin{claim}
\label{claim:2-connected-fullerene}
$G$ is $2$-connected.
\end{claim}
\claimproof
Note that the proof could be done in the same lines as for Claim~\ref{clm:2-connected-K4minorfree} for exact square 4-colorability of $K_4$-minor-free graphs. However, having five colors, we give a simpler proof here.

Suppose that $G$ has a cut-vertex $v$. Since $G$ is subcubic, it has a bridge $uv$. If one of $u$ or $v$ is of degree at most 2 (say, it is $u$), then an exact square 5-coloring of $G'=G-u$ extends to $G$ by using a proper permutation of colors in one of the connected components of $G'$.

Therefore, both $u$ and $v$ are 3-vertices and the graph $G'=G-uv$ has exactly two connected components $G_u$ and $G_v$, containing $u$ and $v$ respectively. Let $u_1$ and $u_2$ (resp. $v_1$ and $v_2$) be the neighbors of $u$ (resp. $v$) in $G_u$ (resp. $G_v$).
By minimality of $G$, graphs $G_u$ and $G_v$ are exact square 5-colorable independently. Take such a coloring $\phi^1$ of $G_u$ and let $\phi^1(u_1)=1$, $\phi^1(u_2)=2$ and $\phi^1(u)=\alpha$. We show that an exact square 5-coloring $\phi^2$ of $G_v$ can be chosen to be compatible with $\phi^1$ in $G$. Without loss of generality we can fix $\phi^2(v)=3\neq\alpha$. Then by applying a proper permutation of colors of $\phi^2$ on $G_v$, one can choose $\phi^2(v_1)$ and $\phi^2(v_2)$ such that $\alpha\notin\{\phi^2(v_1),\phi^2(v_2)\}$ and we are done.\hfill\smallqed\medskip

\begin{claim}\label{claim:distance4_2v}
Any two $2$-vertices of $G$ are at distance at least~$4$.
\end{claim}
\claimproof
The configurations of Figure~\ref{fig:2-verticesAtDistance4} are reducible. Indeed, if one of these configurations occurs, then we remove from $G$ vertices $y,v$ (resp. $y,z,v$ and $y,z,t,v$) in the case of Configuration~\ref{subfig:2-verticesAtDistance1} (resp. \ref{subfig:2-verticesAtDistance2} and \ref{subfig:2-verticesAtDistance3}), in order to obtain a graph $G'$. The graph $G'$ is in $\mathcal C$ and, therefore, has an exact square 5-coloring which can be easily extended to $G$.
\hfill\smallqed\medskip

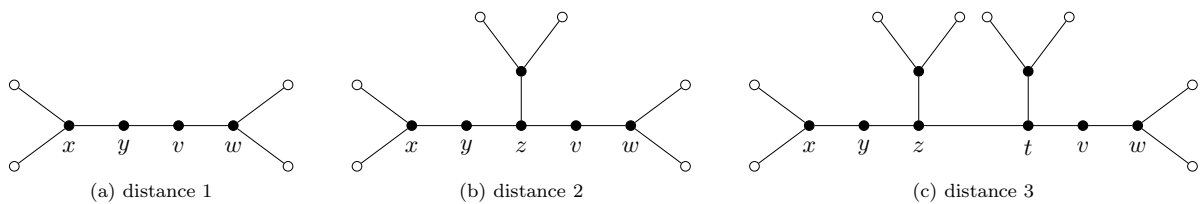
\begin{figure}[!ht]
\centering\scalebox{.9}{
\subfloat[distance 1]{\label{subfig:2-verticesAtDistance1}
\begin{tikzpicture}[join=bevel,inner sep=0.5mm,scale=0.8]
  \pgfmathsetmacro\decname{-0.4};
  \node (x1) at (0,0.75) [draw, circle, fill=white] {};
  \node (x2) at (0,-0.75) [draw, circle, fill=white] {};
  \node (x) at (1,0) [draw, circle, fill=black] {};
  \node  at (1,\decname) {$x$};
  \node (y) at (2,0) [draw, circle, fill=black] {};
  \node  at (2,\decname) {$y$};
  \node (v) at (3,0) [draw, circle, fill=black] {};
  \node  at (3,\decname) {$v$};
  \node (w) at (4,0) [draw, circle, fill=black] {};
  \node  at (4,\decname) {$w$};
  \node (w1) at (5,0.75) [draw, circle, fill=white] {};
  \node (w2) at (5,-0.75) [draw, circle, fill=white] {};
  \draw [-] (w) -- (w1);
  \draw [-] (w) -- (w2);
  \draw [-] (x1) -- (x) -- (x2);
  \draw [-] (x) -- (w);
\end{tikzpicture}
}\hspace*{0.5cm}
\subfloat[distance 2]{\label{subfig:2-verticesAtDistance2}
\begin{tikzpicture}[join=bevel,inner sep=0.5mm,scale=0.8]
  \pgfmathsetmacro\decname{-0.4};
  \node (x1) at (0,0.75) [draw, circle, fill=white] {};
  \node (x2) at (0,-0.75) [draw, circle, fill=white] {};
  \node (x) at (1,0) [draw, circle, fill=black] {};
  \node  at (1,\decname) {$x$};
  \node (y) at (2,0) [draw, circle, fill=black] {};
  \node  at (2,\decname) {$y$};
  \node (z) at (3,0) [draw, circle, fill=black] {};
  \node  at (3,\decname) {$z$};
  \node (v) at (4,0) [draw, circle, fill=black] {};
  \node  at (4,\decname) {$v$};
  \node (w) at (5,0) [draw, circle, fill=black] {};
  \node  at (5,\decname) {$w$};
  \node (z1) at (3,1) [draw, circle, fill=black] {};
  \node (z2) at (2.25,2) [draw, circle, fill=white] {};
  \node (z3) at (3.75,2) [draw, circle, fill=white] {};
  \node (w1) at (6,0.75) [draw, circle, fill=white] {};
  \node (w2) at (6,-0.75) [draw, circle, fill=white] {};
  \draw [-] (w) -- (w1);
  \draw [-] (w) -- (w2);
  \draw [-] (x1) -- (x) -- (x2);
  \draw [-] (z2) -- (z1) -- (z3);
  \draw [-] (x) -- (w);
  \draw [-] (z) -- (z1);
\end{tikzpicture}
}\hspace*{0.5cm}
\subfloat[distance 3]{\label{subfig:2-verticesAtDistance3}
\begin{tikzpicture}[join=bevel,inner sep=0.5mm,scale=0.8]
  \pgfmathsetmacro\decname{-0.4};
  \node (x1) at (0,0.75) [draw, circle, fill=white] {};
  \node (x2) at (0,-0.75) [draw, circle, fill=white] {};
  \node (x) at (1,0) [draw, circle, fill=black] {};
  \node  at (1,\decname) {$x$};
  \node (y) at (2,0) [draw, circle, fill=black] {};
  \node  at (2,\decname) {$y$};
  \node (z) at (3,0) [draw, circle, fill=black] {};
  \node  at (3,\decname) {$z$};
  \node (t) at (5,0) [draw, circle, fill=black] {};
  \node  at (5,\decname) {$t$};
  \node (v) at (6,0) [draw, circle, fill=black] {};
  \node  at (6,\decname) {$v$};
  \node (w) at (7,0) [draw, circle, fill=black] {};
  \node  at (7,\decname) {$w$};
  \node (z1) at (3,1) [draw, circle, fill=black] {};
  \node (z2) at (2.25,2) [draw, circle, fill=white] {};
  \node (z3) at (3.75,2) [draw, circle, fill=white] {};
  \node (t1) at (5,1) [draw, circle, fill=black] {};
  \node (t2) at (4.25,2) [draw, circle, fill=white] {};
  \node (t3) at (5.75,2) [draw, circle, fill=white] {};
  \node (w1) at (8,0.75) [draw, circle, fill=white] {};
  \node (w2) at (8,-0.75) [draw, circle, fill=white] {};
  \draw [-] (w) -- (w1);
  \draw [-] (w) -- (w2);
  \draw [-] (x1) -- (x) -- (x2);
  \draw [-] (t2) -- (t1) -- (t3);
  \draw [-] (z2) -- (z1) -- (z3);
  \draw [-] (x) -- (w);
  \draw [-] (z) -- (z1);
  \draw [-] (t) -- (t1);
\end{tikzpicture}
}
}
\caption{Configurations where two 2-vertices are at distance at most 3. The neighborhood of the black vertices is exactly the one depicted in the figure.}
\label{fig:2-verticesAtDistance4}
\end{figure}

\begin{claim}\label{clm:9faces}
$G$ has no $9$-face.
\end{claim}
\claimproof
By Lemma~\ref{lem:small_cycle_fullerene}, every $9$-face of $G$ must contain three $2$-vertices pairwise at distance at most~$3$, contradicting Claim~\ref{claim:distance4_2v}.
\hfill\smallqed\medskip

For a face $f$ of $G$, let $\ell(f)$ denote the length of $f$, and $n_2(f)$ the number of $2$-vertices on the boundary of $f$.

\begin{claim}\label{obs:n_2v}
For every face $f$ of $G$, we have then $n_2(f)\le\left\lfloor\frac{\ell(f)}{4}\right\rfloor$.
\end{claim}
\claimproof
This follows directly from Claim~\ref{claim:distance4_2v}.
\hfill\smallqed\medskip

\begin{claim}
$G$ is a fullerene graph.
\end{claim}
\claimproof
Let $F(G)$ denote the set of faces of $G$ in its planar embedding. By Euler's Formula we have the following:
\begin{equation}\label{eq01}
\sum_{v \in V(G)}\,(2d(v)-6)\,+\,\sum_{f \in F(G)}\,(\ell(f)-6)\,=\,-12
\end{equation}

We assign to each vertex $v$ the charge $\omega(v)=2d(v)-6$ and to each face $f$ the charge $\omega(f)=\ell(f)-6$. We redistribute the charges by applying the following rule: every face $f$ gives 1 to each 2-vertex lying on its boundary.

Note that after this redistribution of charges, the initial sum of charges is preserved. We analyse the new amount of charges of vertices and faces of $G$:

\begin{itemize}
\item every $3$-vertex has charge $0$,
\item every $2$-vertex has charge $0$ since by Claim~\ref{claim:2-connected-fullerene} it lies on two faces,
\item every face $f$ with $\ell(f)\geq 10$ has strictly positive charge by Claim~\ref{obs:n_2v},
\item every face $f$ with $\ell(f)=6$ has charge $0$ by Claim~\ref{obs:n_2v},
\item every face $f$ with $\ell(f)=5$ has charge $-1$ by Claim~\ref{obs:n_2v}.
\end{itemize}

Therefore, in order to obtain a total charge of $-12$ after the redistribution of charges, we conclude that $G$ contains exactly twelve 5-faces and no other face of length at least 10. On the other hand, by definition of $\mathcal{C}$, we know that $G$ has no $7$-faces, nor $8$-faces. Also by Claim~\ref{clm:9faces}, $G$ has no $9$-faces. Thus we conclude that $G$ is a fullerene graph.\hfill\smallqed\medskip

\begin{claim}\label{claim:no-surrounded-6-faces}
  Every $6$-face of $G$ is adjacent to at least one $5$-face.
\end{claim}
\claimproof Suppose to the contrary that $G$ contains a $6$-face adjacent to six $6$-faces, as in Figure~\ref{subfig:6_6-face}. The graph $G_1=G-\{a,b,c,\ldots,x\}$ belongs to $\mathcal{C}$ and is exact square  5-colorable. By applying such a coloring to $G$, we get a valid partial exact square coloring of $G$. Observe that the remaining uncolored vertices induce two isomorphic connected components in $G\sharptwo$ (see Figure~\ref{subfig:triangulations}). Thus, each can be (properly) colored independently. By counting the number of remaining colors for each of the vertices of each of these connected components we obtain the configuration of Lemma~\ref{lemma:list_col_triangulation}, so we are done.
\hfill\smallqed\medskip

\begin{figure}[!ht]
\centering
\subfloat[A 6-face surrounded by 6-faces]{
\label{subfig:6_6-face}
\scalebox{0.9}{
	\begin{tikzpicture}
	\tikzstyle{vertex}=[draw, circle, fill=white, inner sep=0.55mm]
	  \foreach \i in {0,...,1} {
		\foreach \a in {0,120,-120} {
		  \draw (3*\i,2*sin{60}*\i) -- +(\a:1);
		  \draw (3*\i+3*cos{60},2*sin{60}*\i+sin{60}) -- +(\a:1);
		  \draw (3*\i,4*sin{60}) -- +(\a:1);
		  \draw (0,2*sin{60}) -- +(\a:1);
		  \draw (3,0) -- +(\a:1);
		}
	  }
	  
	  \foreach \a in {0,120,-120} {
		\draw (3*cos{60},3*sin{60}) -- +(\a:1);
		\draw (3*cos{60},5*sin{60}) -- +(\a:1);
		\draw (7*cos{60}+1,sin{60}) -- +(\a:1);
		\draw (3*cos{60},-1*sin{60}) -- +(\a:1);
	  }
	  
	  \draw (-1*cos{60},3*sin{60}) -- +(0:-1);
	  \draw (-1*cos{60},sin{60}) -- +(0:-1);
	  \draw (3*cos{60}+1,-1*sin{60}) -- +(-60:1);
	  \draw (4,0) -- +(-60:1);
	  \draw (4,4*sin{60}) -- +(60:1);
	  \draw (3*cos{60}+1,5*sin{60}) -- +(60:1);
	  
	  \node at (3*cos{60},5*sin{60}) [vertex] {a};
	  \node at (3*cos{60}+1,5*sin{60}) [vertex] {b};
	  \node at (0,4*sin{60}) [vertex] {c};
	  \node at (1,4*sin{60}) [vertex] {d};
	  \node at (3,4*sin{60}) [vertex] {e};
	  \node at (4,4*sin{60}) [vertex] {f};
	  \node at (-1*cos{60},3*sin{60}) [vertex] {g};
	  \node at (3*cos{60},3*sin{60}) [vertex] {h};
	  \node at (3*cos{60}+1,3*sin{60}) [vertex] {i};
	  \node at (7*cos{60}+1,3*sin{60}) [vertex] {j};
	  \node at (0,2*sin{60}) [vertex] {k};
	  \node at (1,2*sin{60}) [vertex] {l};
	  \node at (3,2*sin{60}) [vertex] {m};
	  \node at (4,2*sin{60}) [vertex] {n};
	  \node at (-1*cos{60},sin{60}) [vertex] {o};
	  \node at (3*cos{60},sin{60}) [vertex] {p};
	  \node at (3*cos{60}+1,sin{60}) [vertex] {q};
	  \node at (7*cos{60}+1,sin{60}) [vertex] {r};
	  \node at (0,0) [vertex] {s};
	  \node at (1,0) [vertex] {t};
	  \node at (3,0) [vertex] {u};
	  \node at (4,0) [vertex] {v};
	  \node at (3*cos{60},-1*sin{60}) [vertex] {w};
	  \node at (3*cos{60}+1,-1*sin{60}) [vertex] {x};
	\end{tikzpicture}
	}
}\hspace*{0.4cm}
\subfloat[Exact square]{
\label{subfig:triangulations}
\scalebox{0.9}{
	\begin{tikzpicture}[scale=0.85]
	\tikzstyle{vertex}=[draw, circle, inner sep=0.55mm]
	
	  \node (a) at (3*cos{60},5*sin{60}) [vertex] {a};
	  \node (c) at (0,4*sin{60}) [vertex] {c};
	  \node (e) at (3,4*sin{60}) [vertex] {e};
	  \node (h) at (3*cos{60},3*sin{60}) [vertex] {h};
	  \node (j) at (7*cos{60}+1,3*sin{60}) [vertex] {j};
	  \node (k) at (0,2*sin{60}) [vertex] {k};
	  \node (m) at (3,2*sin{60}) [vertex] {m};
	  \node (p) at (3*cos{60},sin{60}) [vertex] {p};
	  \node (r) at (7*cos{60}+1,sin{60}) [vertex] {r};
	  \node (s) at (0,0) [vertex] {s};
	  \node (u) at (3,0) [vertex] {u};
	  \node (w) at (3*cos{60},-1*sin{60}) [vertex] {w};
	  
	  \draw (a)--(c)--(k)--(s)--(w)--(u)--(r)--(j)--(e)--(a);
	  \draw (c)--(h)--(a) (k)--(p)--(h) (s)--(p)--(w)
	  		(m)--(p)--(u) (h)--(m)--(e) (j)--(m)--(r)
	  		(m)--(u) (k)--(h)--(e);
	\end{tikzpicture}\hspace*{0.4cm}
	\begin{tikzpicture}[scale=0.85]
	\tikzstyle{vertex}=[draw, circle, inner sep=0.55mm]
	
	  \node (b) at (3*cos{60}+1,5*sin{60}) [vertex] {b};
	  \node (d) at (1,4*sin{60}) [vertex] {d};
	  \node (f) at (4,4*sin{60}) [vertex] {f};
	  \node (g) at (-1*cos{60},3*sin{60}) [vertex] {g};
	  \node (i) at (3*cos{60}+1,3*sin{60}) [vertex] {i};
	  \node (l) at (1,2*sin{60}) [vertex] {l};
	  \node (n) at (4,2*sin{60}) [vertex] {n};
	  \node (o) at (-1*cos{60},sin{60}) [vertex] {o};
	  \node (q) at (3*cos{60}+1,sin{60}) [vertex] {q};
	  \node (t) at (1,0) [vertex] {t};
	  \node (v) at (4,0) [vertex] {v};
	  \node (x) at (3*cos{60}+1,-1*sin{60}) [vertex] {x};
	  
	  \draw (b)--(d)--(g)--(o)--(t)--(x)--(v)--(n)--(f)--(b);
	  \draw (b)--(i)--(d) (l)--(i)--(q) (l)--(q)--(t)
	  		(g)--(l)--(o) (l)--(t)--(o) (t)--(x)--(q)
	  		(q)--(i)--(n) (v)--(q)--(n) (l)--(d) (i)--(f);
	\end{tikzpicture}
	}
}
\caption{A 6-face surrounded only by 6-faces and its exact square.}
\label{fig:6_6-face}
\end{figure}
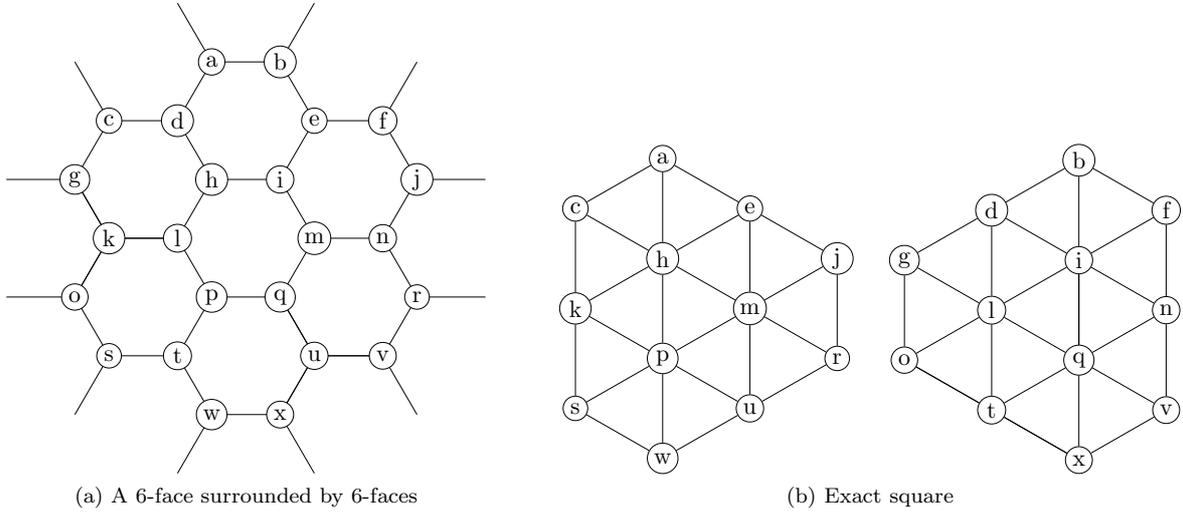

\begin{claim}\label{claim:no-6_5-face}
  Every $6$-face of $G$ is adjacent to at least two $5$-faces.
\end{claim}
\claimproof  Suppose to the contrary that $G$ contains a $6$-face adjacent to at least five $6$-faces. By Claim~\ref{claim:no-surrounded-6-faces}, this 6-face has exactly five adjacent 6-faces and one 5-face, as depicted in Figure~\ref{subfig:6_5-face}. The graph $G_1=G-\{a,b,c,\ldots,q,s,\ldots,x\}$ belongs to $\mathcal{C}$ and is exact square  5-colorable. By applying such a coloring of $G_1$ to $G$, we get a valid partial exact square coloring. Observe that the remaining uncolored vertices induce in $G\sharptwo$ the graph depicted in Figure~\ref{subfig:square_of_6_5-face}. After counting the number of available colors for each uncolored vertex, we properly color this graph using Lemma~\ref{lemma:6-wheel}:
\begin{enumerate}
\item since $|L(v)|\geq 3$ and $|L(t)|\geq 3$, by the pigeonhole principle assign to $v$ and $t$ the same color,
\item color vertices $j,e,a,c$ in this order,
\item color vertices $m,h,k,s,w,u,p$ by Lemma~\ref{lemma:6-wheel},
\item color vertices $x,o,g$ in this order,
\item color vertices $n,q,l,d,b,f,i$ by  Lemma~\ref{lemma:6-wheel}.
\end{enumerate}
Thus, the claim is proved.\hfill\smallqed\medskip

\begin{figure}[!ht]
\centering
\subfloat[A 6-face adjacent to exactly one 5-face]{
\label{subfig:6_5-face}
	\scalebox{0.9}{
		\begin{tikzpicture}
		\tikzstyle{vertex}=[draw, circle, fill=white, inner sep=0.55mm]
		 
		  \foreach \i in {0,...,1} {
			\foreach \a in {0,120,-120} {
			  \draw (3*\i,2*sin{60}*\i) -- +(\a:1);
			  \draw (3*\i+3*cos{60},2*sin{60}*\i+sin{60}) -- +(\a:1);
			  \draw (3*\i,4*sin{60}) -- +(\a:1);
			  \draw (0,2*sin{60}) -- +(\a:1);
			  \draw (3,0) -- +(\a:1);
			}
		  }
		  
		  \foreach \a in {0,120,-120} {
			\draw (3*cos{60},3*sin{60}) -- +(\a:1);
			\draw (3*cos{60},5*sin{60}) -- +(\a:1);
			\draw (3*cos{60},-1*sin{60}) -- +(\a:1);
		  }
		  
		  \draw (4,2*sin{60}) -- (4,0);
		  \draw (-1*cos{60},3*sin{60}) -- +(0:-1);
		  \draw (-1*cos{60},sin{60}) -- +(0:-1);
		  \draw (3*cos{60}+1,-1*sin{60}) -- +(-60:1);
		  \draw (4,0) -- +(-60:1);
		  \draw (4,4*sin{60}) -- +(60:1);
		  \draw (3*cos{60}+1,5*sin{60}) -- +(60:1);
		  
		  \node at (3*cos{60},5*sin{60}) [vertex] {a};
		  \node at (3*cos{60}+1,5*sin{60}) [vertex] {b};
		  \node at (0,4*sin{60}) [vertex] {c};
		  \node at (1,4*sin{60}) [vertex] {d};
		  \node at (3,4*sin{60}) [vertex] {e};
		  \node at (4,4*sin{60}) [vertex] {f};
		  \node at (-1*cos{60},3*sin{60}) [vertex] {g};
		  \node at (3*cos{60},3*sin{60}) [vertex] {h};
		  \node at (3*cos{60}+1,3*sin{60}) [vertex] {i};
		  \node at (7*cos{60}+1,3*sin{60}) [vertex] {j};
		  \node at (0,2*sin{60}) [vertex] {k};
		  \node at (1,2*sin{60}) [vertex] {l};
		  \node at (3,2*sin{60}) [vertex] {m};
		  \node at (4,2*sin{60}) [vertex] {n};
		  \node at (-1*cos{60},sin{60}) [vertex] {o};
		  \node at (3*cos{60},sin{60}) [vertex] {p};
		  \node at (3*cos{60}+1,sin{60}) [vertex] {q};
		  \node at (0,0) [vertex] {s};
		  \node at (1,0) [vertex] {t};
		  \node at (3,0) [vertex] {u};
		  \node at (4,0) [vertex] {v};
		  \node at (3*cos{60},-1*sin{60}) [vertex] {w};
		  \node at (3*cos{60}+1,-1*sin{60}) [vertex] {x};
		\end{tikzpicture}
	}
}\hspace*{0.4cm}
\subfloat[Exact square]{
\label{subfig:square_of_6_5-face}
	\scalebox{0.9}{
		\begin{tikzpicture}[scale=0.85]
		\tikzstyle{vertex}=[draw, circle, inner sep=0.55mm]
		
		  \node (a) at (3*cos{60},5*sin{60}) [vertex] {a};
		  \node (c) at (0,4*sin{60}) [vertex] {c};
		  \node (e) at (3,4*sin{60}) [vertex] {e};
		  \node (h) at (3*cos{60},3*sin{60}) [vertex] {h};
		  \node (j) at (7*cos{60}+1,3*sin{60}) [vertex] {j};
		  \node (k) at (0,2*sin{60}) [vertex] {k};
		  \node (m) at (3,2*sin{60}) [vertex] {m};
		  \node (p) at (3*cos{60},sin{60}) [vertex] {p};
		  \node (s) at (0,0) [vertex] {s};
		  \node (u) at (3,0) [vertex] {u};
		  \node (w) at (3*cos{60},-1*sin{60}) [vertex] {w};
		  
		  \draw (a)--(c)--(k)--(s)--(w)--(u) (j)--(e)--(a);
		  \draw (c)--(h)--(a) (k)--(p)--(h) (s)--(p)--(w)
		  		(m)--(p)--(u) (h)--(m)--(e) (j)--(m)
		  		(m)--(u) (k)--(h)--(e);
		
		  \begin{scope}[shift={(6,0)}]
			  \node (b) at (3*cos{60},5*sin{60}) [vertex] {b};
			  \node (f) at (0,4*sin{60}) [vertex] {f};
			  \node (d) at (3,4*sin{60}) [vertex] {d};
			  \node (i) at (3*cos{60},3*sin{60}) [vertex] {i};
			  \node (g) at (7*cos{60}+1,3*sin{60}) [vertex] {g};
			  \node (n) at (0,2*sin{60}) [vertex] {n};
			  \node (l) at (3,2*sin{60}) [vertex] {l};
			  \node (q) at (3*cos{60},sin{60}) [vertex] {q};
			  \node (o) at (7*cos{60}+1,sin{60}) [vertex] {o};
			  \node (v) at (0,0) [vertex] {v};
			  \node (t) at (3,0) [vertex] {t};
			  \node (x) at (3*cos{60},-1*sin{60}) [vertex] {x};
			  
			  \draw (b)--(f)--(n) (v)--(x)--(t) (b)--(d)--(g)--(o)--(t);
			  \draw (f)--(i)--(b) (n)--(q)--(i) (v)--(q)--(x) (g)--(l)--(o)
			  		(l)--(q)--(t) (d)--(l)--(i) (j)--(v)
			  		(l)--(t) (n)--(i)--(d)
			  		(u)--(n) (m)--(v);
		  \end{scope}
		\end{tikzpicture}
	}
}
\caption{A 6-face surrounded only by one 5-face and its exact square.}
\label{fig:6_5-face}
\end{figure}
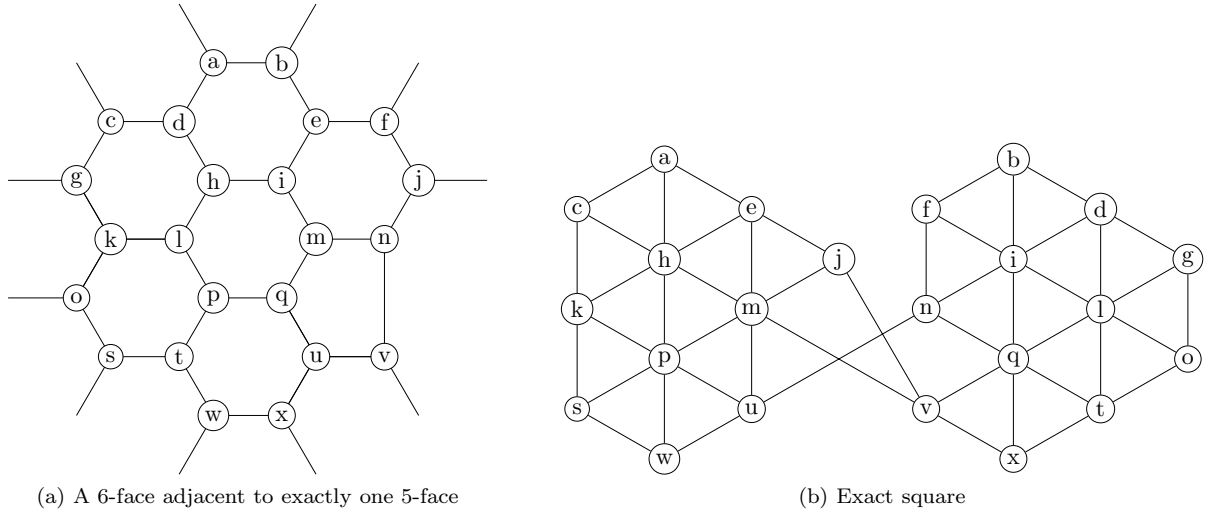

We can now deduce our last claim.

\begin{claim}

$G$ has at most $80$ vertices.
\end{claim}
\claimproof
Since there are exactly twelve $5$-faces in $G$ and each can be adjacent to at most five $6$-faces, there are at most sixty $6$-faces in $G$.
By Claims~\ref{claim:no-surrounded-6-faces} and~\ref{claim:no-6_5-face}, every $6$-face is adjacent to at least two $5$-faces. Thus there can be at most thirty $6$-faces in $G$. Since each vertex belongs to three faces, we conclude that $G$ has at most $\frac{30\cdot 6+12\cdot 5}{3}=80$ vertices.
\hfill\smallqed\medskip

To finish the proof, we have verified by computer using \textit{SageMath}~\cite{Sage} that all fullerene graphs with up to $80$ vertices are exact square  5-colorable (in fact they are exact square 4-colorable). Thus $G$ does not exist, and our theorem is proved. The list of these fullerene graphs is available online\footnote{\url{https://hog.grinvin.org/Fullerenes}} and was generated independently using the two computer programs \textit{fullgen}~\cite{BD97} and \textit{buckygen}~\cite{BGM12,GM15} (these two different programs use two different methods, and thus this list is trusted to be complete).
\end{proof}

\section{Conclusion}\label{sec:conclu}

We have studied the notion of coloring of exact square on some families of subcubic graphs. This fits into a larger frame of studying the chromatic number of exact distance $d$-power of graphs, which has recently got attention. However, this special case is also closely related to the well studied notion of injective coloring. While in exact square coloring the pair of vertices of an edge are allowed to have a same color, in injective coloring this is only allowed for those edges that are not in a triangle. 

A main question studied here is the maximum possible exact square chromatic number of the class of subcubic planar graphs.
Conjecture~\ref{conj} would imply that this is at most 5. There are examples of subcubic planar graphs which need five colors in any injective coloring, but all known such examples contain triangles and are exact square 4-colorable. These examples are built using $K_4^-$ whose vertices must receive four different colors in an injective coloring, however the exact square of this graph has only one edge and can be colored by two colors only. Thus 4 is also a possible answer for our question.

It can be easily checked that a minimum counterexample to the conjecture has no triangle. Thus, as a natural class, we considered cubic planar graphs each of whose faces is either a 5-cycle or 6-cycle. These are are the duals of planar triangulations having only vertices with degree 5 or 6. Known as fullerene graphs they are well studied. For this class of graphs we characterized the ones admitting an exact square 3-coloring, and we proved that they all admit an exact square 5-coloring.

Further evidence that the upper bound of 5 might be replaced by 4 is the fact that the exact square of any bipartite subcubic planar graph is 4-colorable. Our proof of this fact uses the Four Color Theorem. It would be interesting to give an independent proof. On the other hand, the example of an outerplanar subcubic graph whose exact square is 4-chromatic has a fair number of degree 2 vertices. This suggests that one might be able to build examples of subcubic planar graphs whose exact square is not 4-colorable.  

We would like to ask if Theorem~\ref{thm:SP-4-col} can be strengthened by proving the same upper bound of 4 for the injective coloring of $K_4$-minor-free graphs. The upper bound of 5 is proved in \cite{CHRW2012} and the upper bound of 4 on the class of triangle-free $K_4$-minor-free graphs follows from Theorem~\ref{thm:SP-4-col}.

\subsubsection*{Acknowledgements} We thank Franti\v{s}ek Kardo\v{s} for pointing us out the class of $(6,0)$-nanotubes and for insights on fullerene graphs.

This work was supported by the IFCAM project ``\textit{Applications of graph homomorphisms}'' (MA/IFCAM/18/39) and by the ANR project HOSIGRA (ANR-17-CE40-0022). R.N and P.V. were addionally supported by the ANR project DISTANCIA (ANR-17-CE40-0015). S.M. was partially supported by the Overseas Visiting Doctoral Fellowships (OVDF) program, File number: ODF\_2018\_001449, by Science \& Engineering Research Board (SERB).

\end{document}